\documentclass[aoas]{imsart}

\RequirePackage{amsthm,amsmath,amsfonts,amssymb}
\RequirePackage[authoryear]{natbib}
\usepackage{url}
\usepackage{bm}
\usepackage{mathtools}
\usepackage[dvipsnames]{xcolor}
\usepackage{hyperref}
\usepackage{tikz}
\usetikzlibrary{calc, decorations.pathreplacing, arrows.meta}

\startlocaldefs
\newtheorem{proposition}{Proposition}[section]
\newtheorem{theorem}{Theorem}[section]

\newtheorem{lemma}[theorem]{Lemma}
\endlocaldefs

\def\bSig\mathbf{\Sigma}

\makeatletter
\newcommand*\bigcdot{\mathpalette\bigcdot@{.5}}
\newcommand*\bigcdot@[2]{\mathbin{\vcenter{\hbox{\scalebox{#2}{$\m@th#1\bullet$}}}}}
\makeatother

\begin{document}

\begin{frontmatter}
\title{Inference on the state process of periodically inhomogeneous hidden Markov models for animal behavior}
%\title{A sample article title with some additional note\thanksref{T1}}
\runtitle{Inference for periodically inhomogeneous hidden Markov models}
%\thankstext{T1}{A sample of additional note to the title.}

\begin{aug}
%%%%%%%%%%%%%%%%%%%%%%%%%%%%%%%%%%%%%%%%%%%%%%%
%% Only one address is permitted per author. %%
%% Only division, organization and e-mail is %%
%% included in the address.                  %%
%% Additional information can be included in %%
\author[A]{\fnms{Jan-Ole}~\snm{Koslik}\ead[label=e1]{jan-ole.koslik@uni-bielefeld.de}}, 
\author[A]{\fnms{Carlina C.}~\snm{Feldmann}}, %\ead[label=e2]{???@???}},
\author[A]{\fnms{Sina}~\snm{Mews}}, %\ead[label=e2]{???@???}},
\author[A]{\fnms{Rouven}~\snm{Michels}}, %\ead[label=e2]{???@???}},
\and
\author[A]{\fnms{Roland}~\snm{Langrock}}%\ead[label=e3]{???@???}}
%%%%%%%%%%%%%%%%%%%%%%%%%%%%%%%%%%%%%%%%%%%%%%
%% Addresses                                %%
%%%%%%%%%%%%%%%%%%%%%%%%%%%%%%%%%%%%%%%%%%%%%%
\address[A]{Department of Business Administration and Economics, Bielefeld University, 33615 Bielefeld, Germany\printead[presep={,\ }]{e1}}

%\address[B]{???\printead[presep={,\ }]{e2,e3}}
\end{aug}

\begin{abstract}
\noindent
Over the last decade, hidden Markov models (HMMs) have become increasingly popular in statistical ecology, where they constitute natural tools for studying animal behavior based on complex sensor data.
Corresponding analyses sometimes explicitly focus on --- and in any case need to take into account --- periodic variation, for example by quantifying the activity distribution over the daily cycle or seasonal variation such as migratory behavior.
For HMMs including periodic components, we establish important mathematical properties that allow for comprehensive statistical inference related to periodic variation, thereby also providing guidance for model building and model checking.
Specifically, we derive the periodically varying unconditional state distribution as well as the time-varying and overall state dwell-time distributions 
--- all of which are of key interest when the inferential focus lies on the dynamics of the state process.
Crucially, we demonstrate that the dwell-time distributions of periodically inhomogeneous HMMs can deviate substantially from a geometric distribution, thus compensating biologically unrealistic consequences of the Markov property.
We use the novel inference and model-checking tools to investigate changes in the diel activity patterns of fruit flies in response to changing light conditions.
\end{abstract}


\end{frontmatter}
%%%%%%%%%%%%%%%%%%%%%%%%%%%%%%%%%%%%%%%%%%%%%%
%% Please use \tableofcontents for articles %%
%% with 50 pages and more                   %%
%%%%%%%%%%%%%%%%%%%%%%%%%%%%%%%%%%%%%%%%%%%%%%
%\tableofcontents

%%%%%%%%%%%%%%%%%%%%%%%%%%%%%%%%%%%%%%%%%%%%%%
%%%% Main text entry area:

\section{Introduction}
\label{s:intro}

Over the last two decades, advances in biologging technology have revolutionized behavioral ecology \citep{hussey2015aquatic,kays2015terrestrial}. 
Complex sensor data, nowadays often collected at very high resolution (e.g.\ 1 Hz), allow ecologists to more precisely identify foraging and other movement maneuvers employed by animals, to reveal their interaction with conspecifics and prey, and ultimately to infer how they cope with environmental and anthropogenic change % (??) 
\citep{nathan2022big}. 
In particular, driven by these technological advancements, periodic variation in animal behavior 
can now be studied under natural conditions.
This allows for comprehensive inference focused, for example, on the identification of diurnal rhythms, the quantification of individual heterogeneity in diel variation, and the prediction of seasonal patterns or events such as migratory behavior \citep{hertel2017case,weegman2017using,jannetti2019day,beumer2020application}. % 1-2 more REFs here, which fit to the examples
Even if not of primary interest, ignoring such periodic variation can invalidate statistical inference: standard errors might be underestimated due to residual autocorrelation, and the model formulation as guided by information criteria may be more complex than necessary to compensate for the model misspecification \citep{li2017incorporating,pohle2017selecting}.

A statistical framework that naturally lends itself to inference on the dynamics of animal behavior in general, and periodic variation therein in particular, 
is given by the class of hidden Markov models (HMMs). 
In HMMs, the movement metrics observed --- for instance, the step lengths and turning angles between successive GPS fixes, or the dynamic body acceleration calculated from acceleration sensors --- are regarded as noisy measurements of the underlying, serially correlated behavioral process of the animal \citep{langrock2012flexible,leos2017analysis,mcclintock2020uncovering}. 
Ecological inference then mostly focuses on this unobserved behavioral state process, which is typically modeled as a finite-state Markov chain. % the latter,  
To account for periodic variation in animal behavior, the state-switching probabilities are commonly modeled as functions of time, for example by specifying trigonometric base functions with the desired wavelength as covariates using a multinomial logistic regression framework \citep{li2017incorporating,patterson2017statistical,feldmann2023}. 

For time-homogeneous Markov chains, implications of the estimated state-switching probabilities for behavioral time budgets can conveniently be characterized by (i) the stationary distribution, that is the unconditional distribution of the states, and (ii) the state dwell-time distributions. 
Regarding (i), when including periodic effects in the state process, inference on temporal variation in state occupancy is hampered by the fact that an inhomogeneous Markov chain does not have a stationary distribution.
Instead, an approximate version has commonly been reported in the literature \citep{farhadinia2020understanding,byrnes2021evaluating}, which in general is biased as it ignores the inhomogeneous evolution of the state process. 
Regarding (ii), as a consequence of the Markov property, the state dwell times in homogeneous HMMs are geometrically distributed, implying that the most likely duration of a stay within a state is one time unit, which is often biologically unrealistic.
Inhomogeneity in the state process, and in particular temporal variation, does however affect the dwell-time distributions --- yet to what extent this may alleviate the potentially undesirable characteristics of homogeneous Markov state processes has not been investigated.

In this contribution, we analytically derive the main properties of periodically inhomogeneous Markov state processes, namely (i) the time-varying unconditional distribution of the states and (ii) the state dwell-time distribution (both overall and at fixed times). 
Regarding (i), we demonstrate that the approximate state distribution frequently used as an important summary output in analyses of ecological systems can in fact be severely biased. 
Regarding (ii), we find that the state dwell-time distributions implied by HMMs with periodic components can deviate substantially from a geometric distribution. This highlights that temporal covariates may to some extent compensate for biologically unrealistic consequences of the Markov property. 
Our results do in fact apply to periodically inhomogeneous Markov chains in general --- i.e.\ not only to those that form the state process within an HMM --- but we focus on their role specifically within ecological applications of HMMs, as our research is motivated by the study of temporal niche mechanisms and, more generally, periodic variation in animal behavior, as inferred from noisy sensor data.  
For this type of application, we demonstrate the relevance of our theoretical results in an analysis of time-of-day variation in the behavior of common fruit flies ({\it Drosophila melanogaster}) in response to changing environmental conditions.
The methods derived to compute the summary statistics of interest are implemented in the \texttt{LaMa} \texttt{R} package (\citealt{lama}; see also \citealt{mews2025build}).

\section{Methods}
\label{sec:methods}

\subsection{Hidden Markov models --- definition and notation}

We consider an HMM comprising a state-dependent process $\{X_t\}_{t \in \mathbb{N}}$ (where $X_t$ can be a vector) and a latent state process $\{S_t\}_{t \in \mathbb{N}}$, with $S_t \in \{ 1,\ldots,N\}$ selecting which of $N$ possible component distributions generates $X_t$. 
The state process $\{S_t\}$ is assumed to be a Markov chain of first order, characterized by its initial state distribution and the time-varying transition probability matrix (t.p.m.)
%$$\boldsymbol{\Gamma}^{(t)}=(\gamma_{ij}^{(t)}), \; \text{ with } \; \gamma_{ij}^{(t)}=\\Pr(S_{t}=j | S_{t-i}=i),$$
$$\boldsymbol{\Gamma}^{(t)}=(\gamma_{ij}^{(t)}), \; \text{ with } \; \gamma_{ij}^{(t)}=\Pr(S_{t+1}=j | S_{t}=i), \quad t \in \mathbb{N}.$$
The observed variables $X_t$, $t \in \mathbb{N}$, are assumed to be conditionally independent of each other, given the states.

The methodological development of this contribution is generally applicable but was motivated by ecological applications, where $S_t$ could for example indicate the animal's behavioral state at time $t$  (e.g.\ resting, foraging, traveling), with $X_t$ some noisy measurement of that state (e.g.\ acceleration, movement speed, tortuosity of movement, as commonly recorded by biologgers). 
In these settings it is often necessary to incorporate periodic variation in the state-switching process, for example to account for diurnal rhythms. Associated models and their properties are discussed in the following two sections.  

\subsection{Time-varying state distribution in periodically inhomogeneous Markov state processes} % chains

\label{subsec:pstationary}

We consider a setting with periodically varying state-switching dynamics, such that
\begin{equation}
\label{eq1}
\boldsymbol{\Gamma}^{(t)} = \boldsymbol{\Gamma}^{(t+L)}
\end{equation} 
for $t \in \mathbb{N}$, with $L$ denoting the length of a cycle. 
For ease of notation, we restrict the index $t$ to $\{1, \dotsc, L\}$ corresponding to the $L$ unique matrices.
%As only $L$ unique t.p.m.\ exist, we will sometimes index by $t \in \{1, \dotsc, L\}$ to ease notation.% beginning in $t=1$.

For hourly data and $N=2$, we could for example model time-of-day variation ($L=24$) as 
\begin{equation}
\text{logit} \bigl( \gamma_{ij}^{(t)} \bigr) =  \beta_0^{(ij)} + \beta_1^{(ij)} \sin \Bigl( \frac{2\pi t}{24} \Bigr) + \beta_2^{(ij)} \cos \Bigl( \frac{2\pi t}{24} \Bigr), \; \text{ for } i\neq j. 
\label{eq:sincos} 
\end{equation} 

The interpretation of such transition probabilities as functions of time can be tedious, especially when $N>2$. 
Therefore, it has become common practice 
%--- especially in ecology ---
to instead consider a summary statistic, namely the periodically varying (unconditional) distribution of the states. 
For a given time $t$, this distribution is usually approximated by the hypothetical stationary distribution that would emerge if the process followed transition dynamics that would be implied when holding $\boldsymbol{\Gamma}^{(t)}$ constant over time, i.e.\ the solution to $\boldsymbol{\rho}^{(t)}=\boldsymbol{\rho}^{(t)}\boldsymbol{\Gamma}^{(t)}$ for each $t = 1, \ldots, L$,  subject to $\sum_{i=1}^N \rho_i^{(t)}=1$ \citep{patterson2009}. 
The approximation 
%of $\boldsymbol{\delta}^{(t)}$ 
will in general be biased because it ignores the preceding process dynamics as implied by $\boldsymbol{\Gamma}^{(t-1)},\boldsymbol{\Gamma}^{(t-2)},\ldots, \boldsymbol{\Gamma}^{(t-L)},$ and instead pretends that the process has been following the dynamics as implied by a constant $\boldsymbol{\Gamma}^{(t)}$ for a considerable time.  

However, for periodically inhomogeneous Markov chains as defined in Equation~\eqref{eq1}, there is in fact no need for such an approximation. To see this, consider for fixed $t$ the thinned Markov chain 
$\{S_{t+kL}\}_{k \in \mathbb{N}}$,
%$S_t,S_{t+L},S_{t+2L},\ldots$, 
which is homogeneous with constant t.p.m.\
$$
\tilde{\boldsymbol{\Gamma}}_t = \boldsymbol{\Gamma}^{(t)} \boldsymbol{\Gamma}^{(t+1)} \ldots \boldsymbol{\Gamma}^{(t+L-1)}.
$$
Visual intuition for the homogeneity of $\Tilde{\Gamma}_t$ is given in Figure~\ref{fig:perstat}.

Provided that this thinned Markov chain is irreducible, it has a unique stationary distribution $\boldsymbol{\delta}^{(t)}$, which is the solution to 
\begin{equation} 
\label{eq: deltat}
\boldsymbol{\delta}^{(t)} = \boldsymbol{\delta}^{(t)} \tilde{\boldsymbol{\Gamma}}_t
\end{equation}
(\citealp{Ge2006, Kargapolov2012, touron2019consistency}). 
For large $N$ or large $L$, it will typically be most convenient to calculate $\boldsymbol{\delta}^{(t)}$ recursively for $t=1,\ldots,L$ (see Supplementary Material (S.2)).
In both cases, the computational cost is relatively minor compared to a single likelihood evaluation.
If the initial state distribution of the actual state process is the stationary distribution of the thinned Markov chain active at the start of the time series, then each of the thinned Markov chains is stationary, such that $\boldsymbol{\delta}^{(t)}$ is the state distribution at any time $t=1,\ldots,L$ of interest. We refer to such a process $\{S_t\}_{t \in \mathbb{N}}$ as a \textit{periodically stationary} Markov chain. Otherwise, provided aperiodicity, the solution to Equation~\eqref{eq: deltat} will typically be a good approximation to the unconditional (marginal) distribution of states at time $t$, as each of the $L$ thinned Markov chains --- i.e.\ $\{S_{t+kL}\}_{k \in \mathbb{N}}$ for $t=1,\ldots, L$ --- converges to its respective stationary distribution $\bm{\delta}^{(t)}$. 
We therefore use the terms (periodically) stationary distribution and unconditional state distribution interchangeably, as these objects effectively coincide. 
Computing $\bm{\delta}^{(t)}$ can be done conveniently using the function \texttt{stationary\_p()} from the \texttt{R} package \texttt{LaMa}, which uses the efficient recursive approach outlined in Supplementary Material (S.2).

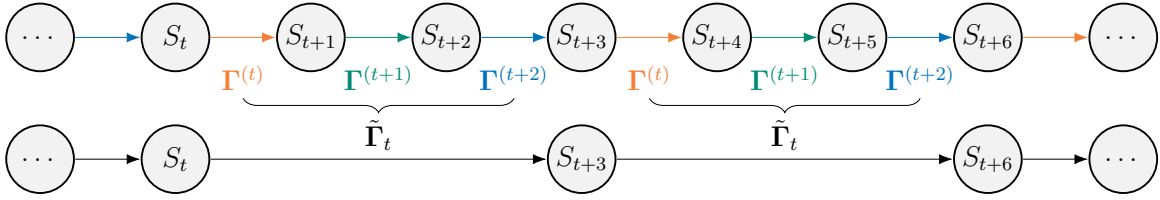
\begin{figure}
\centering
\scalebox{0.8}{
\begin{tikzpicture}
    \coordinate (A) at (0,0);
    \coordinate (B) at (2,0);
    \coordinate (C) at (4,0);
    \coordinate (D) at (6,0);
    \coordinate (E) at (8,0);
    \coordinate (F) at (10,0);
    \coordinate (G) at (12,0);
    \filldraw[fill=black!5, thick] ($(A)-(2,0)$) circle (0.5);
    \draw ($(A)-(2,0)$) node {$\dotsc$};
    \filldraw[fill=black!5, thick] (A) circle (0.5);
    \draw (A) node {$S_{t}$};
    \filldraw[fill=black!5, thick] (B) circle (0.5);
    \draw (B) node {$S_{t+1}$};
    \filldraw[fill=black!5, thick] (C) circle (0.5);
    \draw (C) node {$S_{t+2}$};
    \filldraw[fill=black!5, thick] (D) circle (0.5);
    \draw (D) node {$S_{t+3}$};
    \filldraw[fill=black!5, thick] (E) circle (0.5);
    \draw (E) node {$S_{t+4}$};
    \filldraw[fill=black!5, thick] (F) circle (0.5);
    \draw (F) node {$S_{t+5}$};
    \filldraw[fill=black!5, thick] (G) circle (0.5);
    \draw (G) node {$S_{t+6}$};
    \filldraw[fill=black!5, thick] ($(G)+(2,0)$) circle (0.5);
    \draw ($(G)+(2,0)$) node {$\dotsc$};
    \draw[-{Latex[length=2mm]}, color=RoyalBlue] ($(A)+(-1.48,0)$) -- ($(A)-(0.52,0)$);
    \draw[-{Latex[length=2mm]}, color=Orange] ($(A)+(0.52,0)$) -- ($(B)-(0.52,0)$);
    \draw[-{Latex[length=2mm]}, color=PineGreen] ($(B)+(0.52,0)$) -- ($(C)-(0.52,0)$);
    \draw[-{Latex[length=2mm]}, color=RoyalBlue] ($(C)+(0.52,0)$) -- ($(D)-(0.52,0)$);
    \draw[-{Latex[length=2mm]}, color=Orange] ($(D)+(0.52,0)$) -- ($(E)-(0.52,0)$);
    \draw[-{Latex[length=2mm]}, color=PineGreen] ($(E)+(0.52,0)$) -- ($(F)-(0.52,0)$);
    \draw[-{Latex[length=2mm]}, color=RoyalBlue] ($(F)+(0.52,0)$) -- ($(G)-(0.52,0)$);
    \draw[-{Latex[length=2mm]}, color=Orange] ($(G)+(0.52,0)$) -- ($(G)+(1.48,0)$);
    \draw[color=Orange] ($(A)+(1,-.6)$) node {$\bm{\Gamma}^{(t)}$};
    \draw[color=PineGreen] ($(B)+(1,-.6)$) node {$\bm{\Gamma}^{(t+1)}$};
    \draw[color=RoyalBlue] ($(C)+(1,-.6)$) node {$\bm{\Gamma}^{(t+2)}$};
    \draw[color=Orange] ($(D)+(1,-.6)$) node {$\bm{\Gamma}^{(t)}$};
    \draw[color=PineGreen] ($(E)+(1,-.6)$) node {$\bm{\Gamma}^{(t+1)}$};
    \draw[color=RoyalBlue] ($(F)+(1,-.6)$) node {$\bm{\Gamma}^{(t+2)}$};
    \draw[decorate, decoration={brace, amplitude = 6pt}] ($(C) + (1,-0.9)$) -- ($(A) + (1,-0.9)$) node[midway, below, yshift = -5pt]{$\Tilde{\bm{\Gamma}}_t$};
    \draw[decorate, decoration={brace, amplitude = 6pt}] ($(F) + (1,-0.9)$) -- ($(D) + (1,-0.9)$) node[midway, below, yshift = -5pt]{$\Tilde{\bm{\Gamma}}_t$};
    \filldraw[fill=black!5, thick] ($(A)-(2,1.8)$) circle (0.5);
    \draw ($(A)-(2,1.8)$) node {$\dotsc$};
    \filldraw[fill=black!5, thick] ($(A)-(0,1.8)$) circle (0.5);
    \draw ($(A)-(0,1.8)$) node {$S_{t}$};
    \filldraw[fill=black!5, thick] ($(D)-(0,1.8)$) circle (0.5);
    \draw ($(D)-(0,1.8)$) node {$S_{t+3}$};
    \filldraw[fill=black!5, thick] ($(G)-(0,1.8)$) circle (0.5);
    \draw ($(G)-(0,1.8)$) node {$S_{t+6}$};
    \filldraw[fill=black!5, thick] ($(G)-(-2,1.8)$) circle (0.5);
    \draw ($(G)-(-2,1.8)$) node {$\dotsc$};
    \draw[-{Latex[length=2mm]}] ($(A)+(-2,0)+(0.52,-1.8)$) -- ($(A)+(-0.52,-1.8)$);
    \draw[-{Latex[length=2mm]}] ($(A)+(0.52,-1.8)$) -- ($(D)+(-0.52,-1.8)$);
    \draw[-{Latex[length=2mm]}] ($(D)+(0.52,-1.8)$) -- ($(G)+(-0.52,-1.8)$);
    \draw[-{Latex[length=2mm]}] ($(G)+(0.52,-1.8)$) -- ($(G)+(2,0)+(-0.52,-1.8)$);
\end{tikzpicture}
}
\caption{Example visualization of \textit{periodic stationarity} with $L=3$. The thinned Markov chain $S_t, S_{t+3}, S_{t+6}, \dotsc$ has constant t.p.m.\ $\Tilde{\bm{\Gamma}}_t$.}
\label{fig:perstat}
\end{figure}

To demonstrate the potentially severe bias that arises from the use of the hypothetical stationary distribution $\bm{\rho}^{(t)}$, we simulated three periodically inhomogeneous 2-state Markov chains with different parameter sets, choosing a cycle length of $L = 24$. The parameter values are provided in the Supplementary Material (S.5). For each parameter set, we computed both the hypothetical stationary distribution as well as the periodically stationary distribution introduced above. 
To demonstrate that $\bm{\delta}^{(t)}$ actually gives the true unconditional state probabilities at each time of day, we simulated one very long realization of 1000 days (to reduce stochasticity) from each chain and computed the empirical state frequencies at each time of day. 
The results are shown in Figure \ref{fig:simulations_delta}. We find that in all three scenarios, the hypothetical stationary distribution $\bm{\rho}^{(t)}$ gives only a poor approximation of the true marginal state distribution. While there is only a minor shift in the first scenario, the other two scenarios are much more alarming. In the second scenario, the transition probabilities are such that the process is almost homogeneous. This is also confirmed by the true probability of occupying state 1, which is very close to $0.5$ throughout the entire cycle. However, surprisingly, the hypothetical stationary distribution oscillates between almost zero and almost one, indicating a very strong periodicity, which would lead to completely invalid conclusions in real applications. 
This severe bias is introduced because at certain times of the day, when both off-diagonal entries of the t.p.m.\ are very small, one transition probability is nonetheless substantially larger than the other one. Running the chain indefinitely, i.e., computing the limiting distribution $\bm{\rho}^{(t)}$, gradually shifts more and more probability mass to one of the states, yielding a very unbalanced stationary distribution that never actually arises from the true process.
Similarly, in the third scenario, the hypothetical stationary distribution shows a completely different temporal evolution compared to the true marginal state probabilities, which would also invalidate inference in practice.

\begin{figure}
    \centering
    \includegraphics[width=1\linewidth]{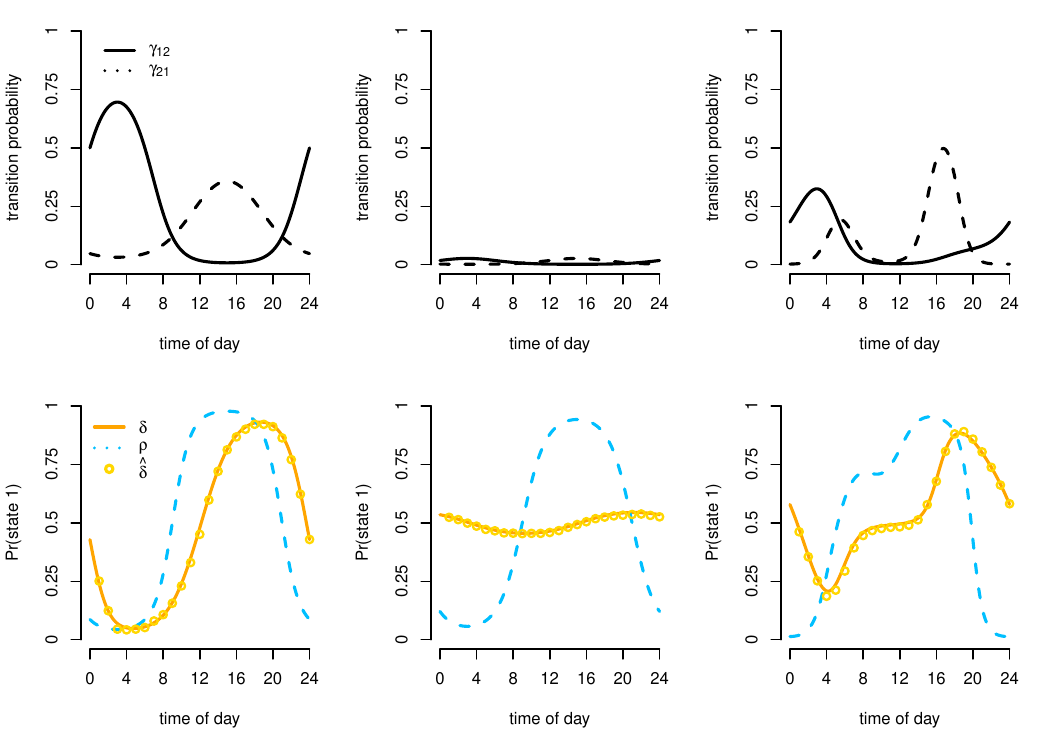}
    \caption{Off-diagonal transition probabilities (top row), implied periodically stationary distribution $\bm{\delta}^{(t)}$, hypothetical distribution $\bm{\rho}^{(t)}$, as well as the empirical marginal state distribution $\Hat{\bm{\delta}}^{(t)}$ (bottom row) for three different trigonometric relationships between the time of day and the transition probabilities (see Supplementary Material (S.5) for the precise model configurations). The empirical distribution is obtained by Monte Carlo simulation based on chains of 1000 days. }
    \label{fig:simulations_delta}
\end{figure}

\subsection{State dwell-time distribution of periodically inhomogeneous Markov chains}

In this section we derive the distribution of the state dwell times implied by a periodically stationary Markov chain. 
We first focus on the \textit{time-varying} state dwell-time distribution, i.e.\ the distribution of the duration of a stay in state $i$ beginning at time $t$, for each $t = 1,\ldots,L$ and each $i = 1, \ldots, N$.

\begin{proposition}
    \label{prop: dwell-time distribution t}
   Consider a periodically inhomogeneous Markov chain defined by $\bm{\Gamma}^{(t)}$, $t = 1,\dotsc,L$. 
   %\textcolor{orange}{Let $R_i^{(t)}$ denote the dwell time in state~$i$ beginning at time $t$. Then, its distribution is defined by the probability mass function}
   For this Markov chain, the probability mass function of the time-varying state dwell-time distribution of a stay in state $i$ beginning at time $t$ is
    \begin{equation}
        d_i^{(t)}(r) = \bigl(1-\gamma_{ii}^{(t+r-1)}\bigr) \prod_{k=1}^{r-1} \gamma_{ii}^{(t+k-1)}, \quad r \in \mathbb{N},
        \label{eq: dwell-time distribution t}
    \end{equation}
where $r$ denotes the length of the stay.
\end{proposition}

For the proof, see Supplementary Material (S.1).
Equation~\eqref{eq: dwell-time distribution t} appears similar to the homogeneous case with geometric dwell-time distributions.
However, here the transition probabilities evolve over time, and as a consequence the probability mass function of this dwell-time distribution is not necessarily monotonically decreasing (cf.\ Figure~\ref{fig: timevarying flies}).
Even though the distribution is not geometric, it does nevertheless exhibit a periodic memorylessness property (see Supplementary Material (S.3)).

The time-varying state dwell-time distributions provide comprehensive information on the state dynamics within a cycle. However, visualization and interpretation of $L$ time-varying distributions for each state $i$ can be tedious (see Supplementary Material (S.4)). 
As an alternative, the means of these distributions may provide a useful summary statistic. 

\begin{proposition}
\label{prop: expected value}
In the setting of Proposition \ref{prop: dwell-time distribution t}, let $R_i^{(t)}$ denote the dwell time in state~$i$ for a stay beginning at time $t$. Then
\begin{equation*}
%    \label{eq: expectation}
    \mathbb{E} R_i^{(t)} = 
    % \sum_{r=1}^\infty r d_i^{(t)}(r) 
    \frac{L + \sum_{r=1}^L r d_i^{(t)}(r)}{\sum_{r=1}^L d_i^{(t)}(r)} - L.
\end{equation*}
\end{proposition}

A proof of this proposition is provided in Supplementary Material (S.1).
Computing the expected dwell time for each state $i = 1, \dotsc, N$ and each time point $t = 1,\dotsc,L$, the variation in the mean dwell times over the course of a cycle can be concisely visualized  (see Figure \ref{fig: expected dwell time}). 

\begin{figure}[!htb]
    \centering
    \includegraphics[width=10cm]{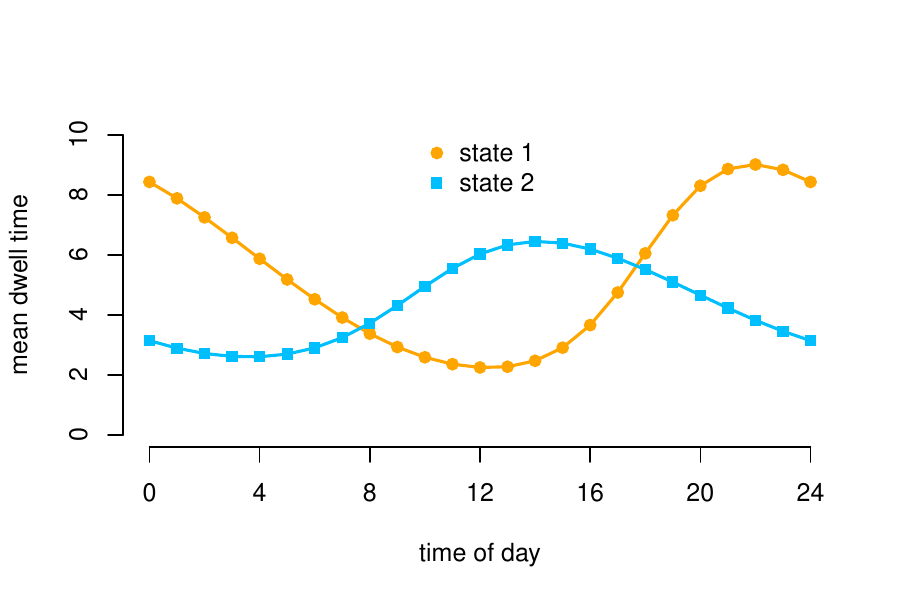}
    \caption{Means of the time-varying state dwell times of an example 2-state HMM with trigonometric modeling of periodic variation (see Supplementary Material (S.5) for the precise model configuration).} % an artificial
    \label{fig: expected dwell time}
\end{figure}

In practical applications, it may be cumbersome to interpret the properties of the time-varying dwell-time distributions, especially when the temporal resolution of the data is high (for example, with minute-by-minute data and diel variation, there would be $1440$ such distributions for each state). 
In addition, the focus of the inference with respect to state dynamics will often be on the \textit{overall} distribution of the durations in the different states, not explicitly conditioning on the start time of the stay. 
We obtain such a distribution as a mixture of the time-varying dwell-time distributions.

\begin{proposition}
\label{prop: dwell-time distribution}
For a periodically stationary Markov chain defined by $\bm{\Gamma}^{(t)}$, $t = 1,\dotsc,L$, the probability mass function of the overall (unconditional) dwell-time distribution in state $i$ is
\begin{equation}
    d_i(r) = \sum_{t=1}^L w_i^{(t)} d_i^{(t)}(r), \qquad r \in \mathbb{N},
    \label{eq: dwell-time distribution}
\end{equation}
with the mixture weights defined as
$$w_i^{(t)} = \frac{\sum_{l \in \mathcal{S} \setminus i} \delta_l^{(t-1)} \gamma_{li}^{(t-1)}}{\sum_{t=1}^L \sum_{l \in \mathcal{S} \setminus i} \delta_l^{(t-1)} \gamma_{li}^{(t-1)}}, \quad t = 1, \dotsc, L,$$
where $\mathcal{S}=\{1, \dotsc, N\}$, $\bm{\Gamma}^{(0)} = \bm{\Gamma}^{(L)}$, $\bm{\delta}^{(0)} = \bm{\delta}^{(L)}$ and $\bm{\delta}^{(t)}$ as in Equation~\eqref{eq: deltat}. Letting $R_i$ denote the \textit{overall} dwell time in state $i$, we further have that
\begin{equation*}
%    \label{eq: expectation}
    \mathbb{E} R_i = \sum_{t=1}^L w_i^{(t)} \mathbb{E} R_i^{(t)} = \sum_{t=1}^L \biggl( w_i^{(t)} \frac{L+\sum_{r=1}^L r d_i^{(t)}(r)}{\sum_{r=1}^L d_i^{(t)}(r)} \biggr) - L.
\end{equation*}
\end{proposition}

The  proof of Proposition~\ref{prop: dwell-time distribution} is given in Supplementary Material (S.1).
It can further be shown that the periodic memorylessness property exhibited by the time-varying dwell-time distribution is inherited by the overall dwell-time distribution (see Supplementary Material (S.3)).

For homogeneous Markov chains, i.e.\ when $\gamma_{ij}^{(t)} = \gamma_{ij}$ for all $t$ and all $i,j \in \mathcal{S}$, the state dwell-time distribution given in Equation~\eqref{eq: dwell-time distribution} simplifies to the geometric case. To see this, we first note that
$$d_i^{(t)}(r) = \bigl(1 - \gamma_{ii}^{(t+r-1)} \bigr) \prod_{k=1}^{r-1} \gamma_{ii}^{(t+k-1)} = (1- \gamma_{ii})\gamma_{ii}^{r-1}$$
is constant in time. Thus,
$$d_i(r) = \sum_{t=1}^L w_i^{(t)} (1- \gamma_{ii})\gamma_{ii}^{r-1} = (1- \gamma_{ii})\gamma_{ii}^{r-1} \sum_{t=1}^L w_i^{(t)} = (1- \gamma_{ii})\gamma_{ii}^{r-1},$$
as the sum of the weights $w_i^{(t)}$ equals one. %\textcolor{orange}{Furthermore the overall dwell-time distribution inherits the memorylessness property from the time-varying dwell-time distribution (for the proof, see Appendix \ref{A2: Memorylessness}).}

If, however, the Markov chain is not homogeneous, then the distribution in Equation~\eqref{eq: dwell-time distribution} can deviate rather substantially from a geometric distribution, and may even be multimodal. 
To illustrate this, Figure~\ref{fig: dwell_distr_sim} displays an example overall state dwell-time distribution implied by an HMM with trigonometric modeling of the periodic variation in the state transition probabilities (see Supplementary Material (S.5) for the model parameters leading to this outcome). 
To verify our theoretical results, we further complemented the exact probability mass function derived in Proposition \ref{prop: dwell-time distribution} by an approximation using  Monte Carlo simulations.
While the latter can always easily be obtained, the exact theoretical result is of course advantageous with respect to both accuracy and computational cost.
Both the time-varying and overall dwell-time distribution can be computed conveniently using the function \texttt{ddwell()} from the \texttt{R} package \texttt{LaMa}. For efficient computation of the time-varying dwell-time distribution, this function recycles all factors in the product when evaluating Equation \eqref{eq: dwell-time distribution t} for a vector of dwell times. The computation of the overall dwell-time distribution then leverages the efficient recursive computation approach for the periodically stationary distribution outlined in Supplementary Material (S.2).

\begin{figure}%[!htb]
    \centering
    \includegraphics[width=14cm]{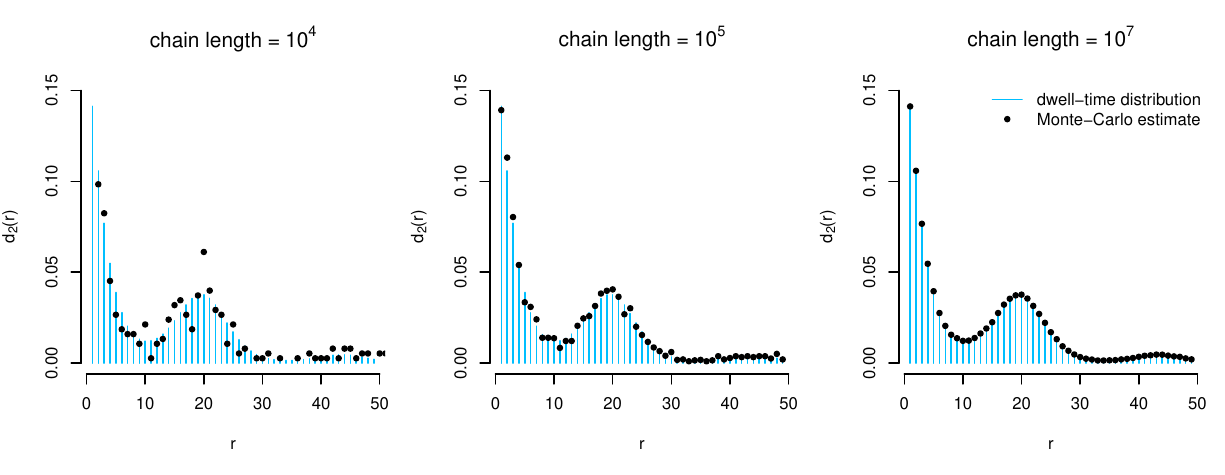}
    \caption{Overall state dwell-time distribution (state 2) of an example 2-state HMM with trigonometric modeling of periodic variation (see Supplementary Material (S.5) for the precise model configuration). Analytical distribution (bars) compared to Monte-Carlo-simulated state dwell times (dots) for different chain lengths.}
    \label{fig: dwell_distr_sim}
\end{figure}

The %analytical calculation of both the 
time-varying and the overall state dwell-time distributions can provide valuable insights into the dynamics of the state process. 
In addition, these results can be used to devise comprehensive model-checking tools for HMMs with periodic variation. 
For example, the model-implied dwell-time distributions can be compared against the empirical state dwell times obtained from state sequences simulated based on locally decoded state probabilities.

\section{Application: Activity of Drosophila melanogaster}

\subsection{Data}

We aim to study the influence of external conditions on the behavior of common fruit flies (\textit{Drosophila melanogaster}) and especially their circadian clock.
Like the majority of animal species, fruit flies restrict their behavioral activity to specific time periods during the 24-hour cycle, a mechanism referred to as the temporal niche.
In particular, synchronizing their circadian clocks to the common light-dark (LD) cycles improves the flies' fitness as they anticipate daily environmental changes \citep{beaver2002loss, bernhardt2020life}. %, which improves their fitness.
To investigate these diel activity patterns, 15 male wild-type fruit flies aged two to three days were entrained under a standard lighting schedule of twelve hours of light followed by twelve hours of darkness (LD condition) for four days. %  a duration of 
Subsequently, the flies were exposed to six consecutive days of uninterrupted darkness (DD condition). 
They were kept in special tubes that tracked their locomotor activity by counting each time a fly passed an infrared beam in the middle of the tube.
% To track the flies' locomotor activity, the study recorded the times each fly passed an infrared beam located in the middle of the tube in which they were kept.
We aggregated these counts into half-hour bins, leading to a total of 6840 observations ranging from 0-300.
%An example of the collected time-series data is shown in Figure~[REF], reflecting the clear diel activity patterns characterized by the fly's anticipation of the light transitions in the morning and evening. 
Boxplots of the activity counts for the different times of day, separated into LD and DD conditions, are shown in Figure~\ref{fig: boxplot}.
The strong diel activity pattern emerges from the flies' anticipation of the light transitions in the morning and evening in the LD condition. 
In constant darkness, the bimodality is much less pronounced as the flies lose these reference points.
We aim to precisely quantify such behavioral differences between the two light conditions, in particular by comprehensively studying the state-switching dynamics leading to the bimodal activity pattern. 

\begin{figure}
    \centering
    \includegraphics[width=14cm]{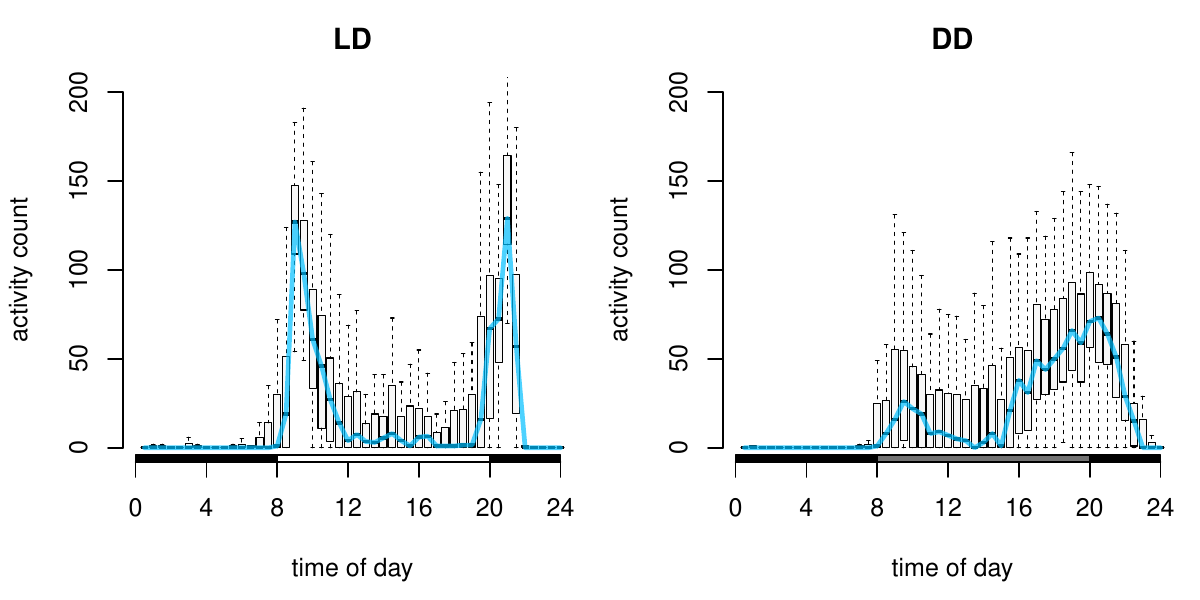}
    \caption{Boxplots of activity counts for each half-hour bin, including all individuals and all days but separated into LD and DD conditions. The thick line represents the median activity for visual clarity.}
    \label{fig: boxplot}
\end{figure}

\subsection{Model formulation}

To investigate the diel activity patterns, % in the motivating fruit fly data from Section~\ref{s:data},
we model the half-hourly activity counts (cf.\ Figure~\ref{fig: boxplot}) using a 2-state HMM with negative binomial state-dependent distributions, as was previously done in \citet{feldmann2023}.
%The state-dependent means are modeled as gamma-distributed random effects to allow for individual-specific differences in the flies' activity levels, while the dispersion parameters are fixed across individuals. % means of the negative binomial distributions
To allow for individual-specific differences in the flies' activity levels, we model the state-dependent means as gamma-distributed random effects, while the dispersion parameters of the negative binomial distributions are shared across individuals.
The state transition probabilities are modeled as functions of the time of day, ensuring % and its interaction with the lighting condition (i.e.\ LD and DD), thus effectively estimating separate state-process parameters for both conditions.
sufficient flexibility for capturing multiple activity peaks throughout the day via the use of trigonometric functions with wavelengths of 24, 12, and 8 hours, i.e.
$$
\text{logit} \bigl( \gamma_{ij}^{(t)} \bigr) =  \beta_0^{(ij)} + \sum_{k=1}^3 \beta_{1k}^{(ij)} \sin \Bigl( \frac{2\pi k t}{48} \Bigr) + \sum_{k=1}^3 \beta_{2k}^{(ij)} \cos \Bigl( \frac{2\pi k t}{48} \Bigr), \quad \text{ for } i\neq j. 
$$
As we are interested in behavioral differences between the lighting schedules, we estimate separate state-process parameters for the LD and DD conditions, respectively. % ...(i.e.\ LD and DD), we effectively estimate separate state-process parameters for both conditions by considering interactions between the lighting schedule and the periodic component. 
To estimate the model parameters, we numerically maximize the joint likelihood computed as the product of the different individuals' likelihoods.
The random effects were marginalized out using numerical integration \citep{schliehe2012application}. 
All models were implemented and fitted in \textbf{R} \citep{R2023} using a parallelized numerical optimization procedure \citep{gerber2019} to speed up the estimation.
The code and data for reproducing all results from this paper can be found at \url{https://github.com/janoleko/Drosophila}.

\subsection{Results}

The fitted HMM distinguishes a low- and a high-activity state for all flies, while allowing for individual differences in their mean activity levels (cf.\ Figure 10 in the Supplementary Material (S.4)).
To investigate the temporal variation in the state dynamics as well as the state dwell times, we apply the inferential tools developed in Section~\ref{sec:methods}. % for periodically inhomogeneous Markov chains.
The time-dependent unconditional state distributions as well as their approximations under the LD and the DD condition are shown in Figure~\ref{fig: deltarho flies}.
% for both the true as well as the approximate distribution often used in the literature
For both light conditions, the observed activity patterns (cf.\ Figure~\ref{fig: boxplot}) are adequately reflected by the true stationary distribution of the inhomogeneous Markov chain, with the activity peaking shortly after the light transitions experienced by the flies in the LD setup. 
While the morning peak in activity is more pronounced under the LD condition, in constant darkness (DD) the flies are active for a longer period of time in the evening hours. 
In contrast to the true distribution $\boldsymbol{\delta}^{(t)}$, the high-activity peaks of the approximation $\boldsymbol{\rho}^{(t)}$ differ in shape and, more importantly, are shifted in time, falsely indicating activity to peak about 1-3 hours earlier (cf.\ also the empirical activity distribution s
hown in Figure~\ref{fig: boxplot}).
Therefore, the approximate version would inevitably result in erroneous conclusions when interested in the exact time and length of flies' activity peaks.

\begin{figure}%[!htb]
    \centering
    \includegraphics[width=14cm]{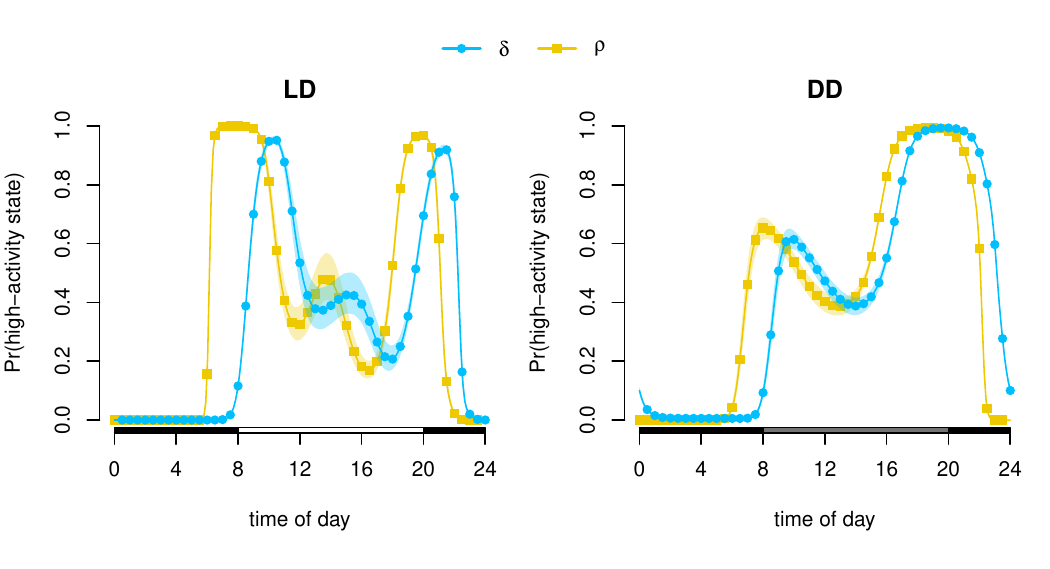}
    \caption{Periodic stationary distribution as a function of the time of day, for LD and DD condition. True stationary distribution (dots) compared to biased approximation (squares). Pointwise confidence intervals were obtained via Monte Carlo simulation from the approximate normal distribution of the maximum likelihood estimator.}
    \label{fig: deltarho flies}
\end{figure}

To further characterize the state dynamics within a cycle, the time-varying dwell-time distributions as derived in Proposition \ref{prop: dwell-time distribution t} can provide temporally fine-grained insights. 
Due to the model accounting for periodic variation in the state-switching dynamics, the model-implied dwell times differ substantially from a geometric distribution, as illustrated in Figure~\ref{fig: timevarying flies}.
Specifically for the morning times when the light was just switched on, the distribution's mode is clearly distinct from one --- in other words, the flies' response to the environmental change is not only to become active, but also to remain active for an extended period of time. During noon and in the early afternoon, activity bouts are more likely to be short. 
Overall, the time-varying dwell-time distributions thus reflect and complement the information gained from the stationary distribution (cf.\ Figure~\ref{fig: deltarho flies}).
However, since it is tedious and hardly feasible to look at all time-dependent dwell-time distributions for both states in the LD and DD condition ($48 \times 2 \times 2 = 192$ in total), it can be useful to instead plot the time-varying means of the distributions as derived in Proposition \ref{prop: expected value} (cf.\ Figure 11 in the Supplementary Material (S.4)).
% However, since it is tiresome and hardly feasible to look at all dwell-time distributions (48 for each state and light condition, in total 192), the time-varying distributional means can be plotted instead (cf.\ Figure~\ref{fig: mean dwell times} in the Appendix).
These mean dwell times summarize the distinct patterns of varying durations in the high- and low-activity state over a day and stress differences between both light conditions.
Solely focusing on the mean can however be accompanied by a substantial loss of information, as the time-varying dwell-time distributions can be multimodal (cf.\ Figure~\ref{fig: timevarying flies}). 

\begin{figure}
    \centering
    \includegraphics[width=14cm]{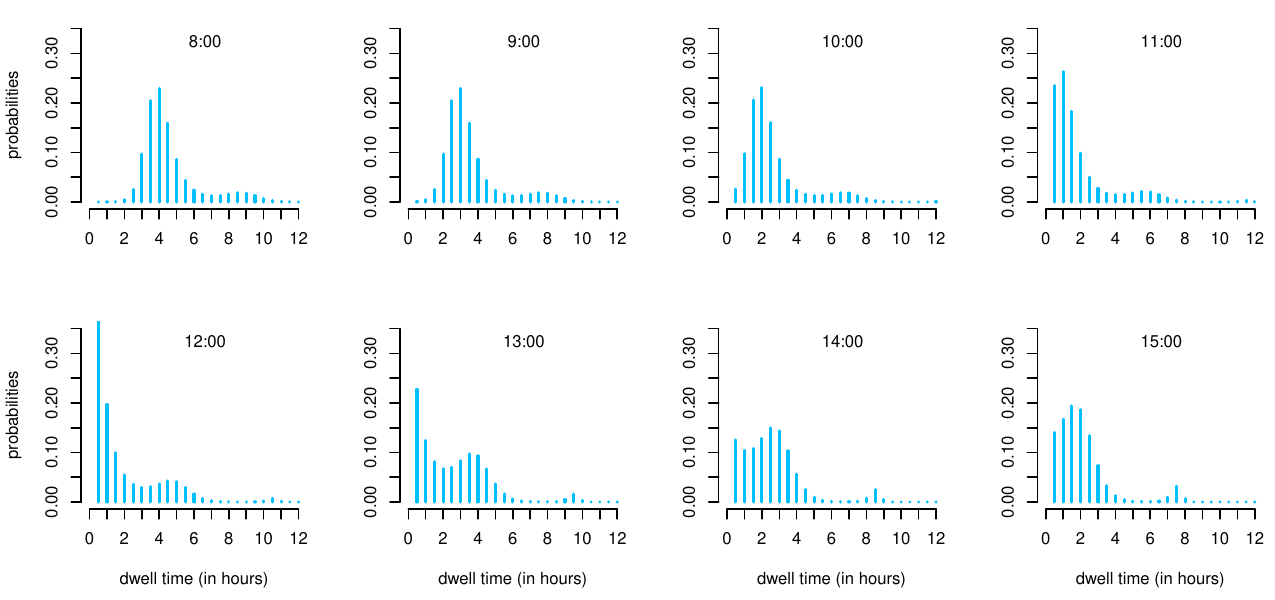}
    \caption{Estimated time-varying dwell-time distributions when initiating high-activity bouts between 8 am and 3 pm under LD condition}
    \label{fig: timevarying flies}
\end{figure}

As an alternative way of looking into the dwell times implied by the model --- not reducing the distributional information but instead %effectively 
averaging over time --- the overall dwell-time distributions as derived in Proposition \ref{prop: dwell-time distribution} are shown in Figure~\ref{fig: overall flies}. These unconditional dwell-time distributions in the active state emphasize overall differences in the state dynamics between the lighting schedules, % characterized by 
specifically showing activity patterns that differ in peak lengths (cf.\ Figure~\ref{fig: deltarho flies}).
In particular, the morning and evening activity peaks in the LD condition are relatively short and about equally long, which is reflected by dwell times mostly below five hours. % corresponding to similar dwell times as reflected in the unimodal overall distribution.
In contrast, the DD condition is characterized by relatively short activity bouts in the morning and notably longer bouts in the evening, leading to a bimodal dwell-time distribution, with considerable mass on lengths of 6--7 hours.
These ecologically interesting differences cannot be uncovered when considering only the distributional means, which are fairly similar in both conditions (LD: 2.7 hours; DD: 2.3 hours).

\begin{figure}[b]
    \centering
    \includegraphics[width=14cm]{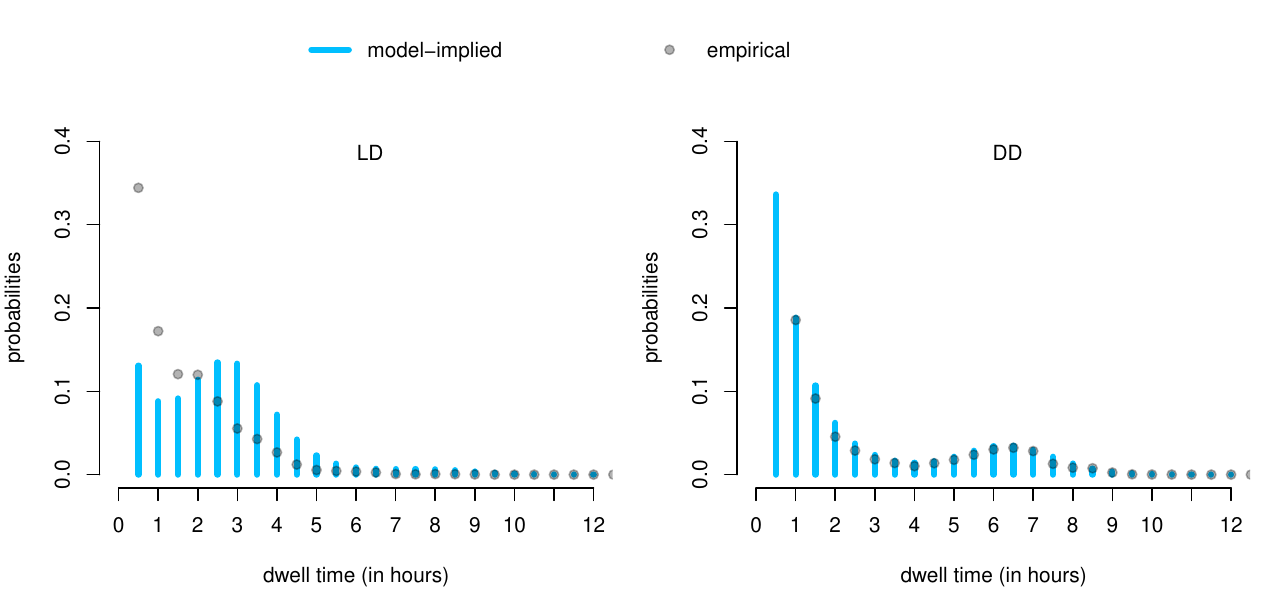}
    \caption{Overall dwell-time distribution of the high-activity state for LD and DD condition, analytically (bars) and empirically (gray dots) derived from the fitted HMM.}
    \label{fig: overall flies}
\end{figure}

In addition to its inferential use, the unconditional dwell-time distribution constitutes a valuable model-checking tool, as it allows for the comparison of the model-implied dwell times to the empirical ones obtained from the locally decoded state sequence.
Specifically, for Figure~\ref{fig: overall flies}, we drew 1000 state sequences based on the locally-decoded state probabilities. For each of these state sequences, we then computed the empirical dwell-time distribution based on a run-length encoding and ultimately averaged these frequencies to obtain the empirical p.m.f.\ depicted in the figure.
In the LD condition, we observe a minor lack of fit between the model-implied and the empirical dwell-time distributions, and a very good fit in the DD condition.
Compared to the model without temporal variation, which restricts the dwell-time distributions to be geometric, the model fit in the DD condition is substantially improved. In contrast, the fit for the LD condition looks slightly worse (see Figure 12 in the Supplementary Material (S.4)).

\subsection{Summary}

To summarize, using periodically inhomogeneous HMMs, we showed that fruit flies exhibit a strong diel pattern in activity under a regular light-dark cycle, which does not vanish but becomes more concentrated on the evening hours when subsequently faced with constant darkness. 
Based on the inferred state dwell-time distributions, the analysis further revealed that dwell times of up to five hours were most likely in the LD condition, whereas dwell times follow a bimodal distribution with mainly very short ($< 2$ hours) or long (6--7 hours) durations in the DD condition. 
These results on common fruit flies' diel patterns and their adaptations to different lighting schedules will serve as a basis for future research to study variation in genetically mutated flies and their temporal niche mechanisms.

\section{Discussion}
\label{s:discuss}

We established several properties of periodically inhomogeneous Markov chains (and hence also corresponding HMMs), specifically the time-varying (unconditional) state distribution, the time-varying and overall state dwell-time distributions, and a % goodness-of-fit test 
model-checking procedure based on the latter. Knowledge of these properties and tools has welcome implications for applied work using corresponding model formulations. 

First, the state distribution as a function of covariates, as displayed for example in Figure~\ref{fig: deltarho flies}, is often a key output of interest in empirical analyses based on HMMs. 
For settings in which the covariate dependence is periodic, for example when modeling diel variation or seasonality, we provide an exact solution to the time-varying state distribution, thus avoiding the need to resort to the approximation commonly considered in the literature \citep{farhadinia2020understanding, byrnes2021evaluating}, which, as demonstrated, can be very poor and hence can lead to invalid inference.

Second, our results concerning the state dwell-time distributions provide practitioners with a valuable new inferential tool. In many HMM-based analyses, the stochastic dynamics of the latent state process are of main interest \citep{jackson2003multistate, mcclintock2020uncovering}. 
In this regard, our results allow users for the first time to consider and report as an important analysis outcome the true model-implied distribution of time spent in the different states, either at fixed times or as a global summary. 
This can reveal insights that would otherwise have been missed, for example the second mode of the dwell-time distribution in the flies' high-activity state when in constant darkness (Figure~\ref{fig: overall flies}).

As a consequence of our findings related to the state dwell-time distributions, we were able to demonstrate that a common criticism expressed towards basic (homogeneous) HMMs, namely that the implied geometric dwell times are often unrealistic, can to a large extent be remedied by inhomogeneous modeling of the state-switching dynamics. 
However, substantial deviations from geometric dwell-time distributions will arise only when the data exhibit considerable periodic variation. 
In some cases, the process to be analyzed involves non-geometric dwell-time distributions that can not be explained by periodic variation or in fact any other measurable covariates, but are instead driven by non-observable internal processes. Hidden semi-Markov models (HSMMs; \citealp{guedon2003estimating}), in which an additional distribution on the positive integers is specified to model the dwell times, may then be effective tools for capturing the dynamics of the state process. HSMMs can also be complemented by (periodic) covariates \citep{koslik2025hidden, lagona2024nonhomogeneous}, thus allowing for the dynamics of the state process --- and consequently the dwell-time distributions --- to be partly explained by external factors. Indeed, for periodic covariates, the periodically stationary state distribution and the overall dwell-time distribution can be derived analogously to Section \ref{sec:methods} for HSMMs (see \citealp{koslik2025hidden}).

Finally, our results allow us to devise a new approach to checking the goodness-of-fit of HMMs involving periodic variation. Specifically,
% Established tools for checking the fit of HMMs have focused primarily on the adequacy of the state-dependent process, assessed for example using pseudo-residuals \citep{buckby2020model}.
a comparison of the model-implied state dwell-time distributions to the ``empirical ones'' obtained from the locally decoded state sequences 
% In contrast, an inspection of the state dwell-time distributions % --- both visually (as in Figure~\ref{fig: overall flies}) but also formally via a chi-squared test --- 
corresponds to a check whether the system's state process is adequately modeled. 
If the state dynamics are indeed of key interest, then they should usually be the central focus of a goodness-of-fit check, and this new type of model check may sometimes be more targeted than the use of pseudo-residuals \citep{buckby2020model}.

The periodic setting considered in this contribution allowed us to analytically derive the summary statistics of interest due to the implied deterministic nature of the covariate (time). 
For the general case, i.e.\ with non-deterministic time-varying covariates such as temperature or precipitation, both the state occupancy and the state dwell-time distribution will depend on the characteristics of the covariate process.
Specifically, both summary statistics are effectively a result of the dynamics of the Markov chain, which in theory depends on the complete history of covariate values, such that the covariates' autocorrelation structure is clearly relevant. 
The natural way to address this would seem to be the formulation and estimation of an additional model for the covariate process, e.g.\ a (vector-)autoregressive process, based on which the desired summary statistics can be obtained using Monte Carlo methods --- we are currently exploring this idea.

% \backmatter

%  This section is optional.  Here is where you will want to cite
%  grants, people who helped with the paper, etc.  But keep it short!

\begin{acks}[Acknowledgments]
The authors are very grateful to Angelica Coculla and Ralf Stanewsky for providing the Drosophila melanogaster activity data. %The authors gratefully acknowledge funding
\end{acks}

\begin{funding}
This research was funded by the German Research Foundation (DFG) as part of the SFB TRR 212 (NC$^3$), project numbers 316099922 and 396782756.
\end{funding}

% \section*{Supporting Information}

% The supplementary material includes proofs for all stated propositions, information on additional properties, and additional figures. Furthermore, the data and code for reproducing all results from the paper can be found at \url{https://github.com/janoleko/Drosophila}.

\begin{supplement}
\stitle{Supplementary Material and R code}
\sdescription{The supplementary material includes proofs for all stated propositions, information on additional properties, additional figures as well as the data and R code for reproducing all results from the paper. It can also be found at \url{https://github.com/janoleko/Drosophila}.}
\end{supplement}

\bibliographystyle{apalike}
\bibliography{refs}

\newpage

\renewcommand{\thesection}{S}

\section{Supplementary material}

\subsection{Proofs}
\label{A: proofs}

\begin{proof}[Proof of Proposition 2.1]
    \begin{align*}0
        d_i^{(t)}(r) &= \Pr(S_{t+r} \neq i, S_{t+r-1} = i, \dotsc, S_{t+1} = i \mid S_t = i, S_{t-1} \neq i) \\
        &= \Pr(S_{t+r} \neq i, S_{t+r-1} = i, \dotsc, S_{t+1} = i \mid S_t = i) \\
        &= \Pr(S_{t+r} \neq i \mid S_{t+r-1} = i) \Pr(S_{t+r-1} = i \mid S_{t+r-2} = i) \dotsc \Pr(S_{t+1} = i \mid S_t = i)\\
        &= (1-\gamma_{ii}^{(t+r-1)}) \prod_{j=1}^{r-1} \gamma_{ii}^{(t+j-1)}
    \end{align*}
\end{proof}
It is evident that the time-varying distribution is identical when the chain has just transitioned into state $i$ or when only conditioning on the chain currently being in state $i$.

\begin{lemma}
    \label{lem: sum_prod}
    Let $\bm{\Gamma}^{(t)}$ be the periodically varying t.p.m.\ of a Markov chain $\{S_t\}$ as defined in Equation (1). Then 
    $$
    %\sum_{r = 1}^L d_i^{(t)}(r) = 
    \sum_{r = 1}^L \bigl(1-\gamma_{ii}^{(t+r-1)}\bigr) \prod_{j = 1}^{r-1}\gamma_{ii}^{(t+j-1)} = 1-\prod_{j = 1}^{L}\gamma_{ii}^{(j)}$$
    for all $L \in \mathbb{N}$. 
\end{lemma}
\begin{proof}
    Without loss of generality, we set $t=1$. We prove the above via induction. \\    
    \textit{Base case: }For $L=1$, it holds that
    $$1 - \gamma_{ii}^{(1)} = 1 - \gamma_{ii}^{(1)}$$
    due to the empty product on the left side.  
    \newpage
    \textit{Induction step:} If the equality holds for $L$,
    %, i.e.\ $\sum_{r = 1}^L \big(1-\gamma_{ii}^{(r)}\big) \prod_{j = 1}^{r-1}\gamma_{ii}^{(j)} = 1-\prod_{j = 1}^{L}\gamma_{ii}^{(j)}$,
    then it also holds for $L+1$ for every $L \in \mathbb{N}$ since
    \begin{align*}
        \sum_{r = 1}^{L+1} \big(1-\gamma_{ii}^{(r)}\big) \prod_{j = 1}^{r-1}\gamma_{ii}^{(j)} &= \big(1-\gamma_{ii}^{(L+1)}\big) \prod_{j=1}^{(L+1)-1}\gamma_{ii}^{(j)} + \sum_{r = 1}^{L} \big(1-\gamma_{ii}^{(r)}\big) \prod_{j = 1}^{k-1}\gamma_{ii}^{(j)} \\
        &= \big(1-\gamma_{ii}^{(L+1)}\big) \prod_{j=1}^{L}\gamma_{ii}^{(j)} + \big(1- \prod_{j = 1}^{L} \gamma_{ii}^{(j)}\big) \\
        &= 1-\prod_{j = 1}^L \gamma_{ii}^{(j)}\big(-1+\gamma_{ii}^{(L+1)}+1\big) \\
        &= 1-\prod_{j = 1}^{L}\gamma_{ii}^{(j)}\gamma_{ii}^{(L+1)} \\
        &= 1-\prod_{j = 1}^{L+1}\gamma_{ii}^{(j)}.
    \end{align*}
    We use the induction condition in the second step, the remaining calculations are straightforward.
\end{proof}

\begin{proof}[Proof of Proposition 2.2]
\label{proof: expected value}
\begin{align*}
    \mathbb{E} R_i^{(t)} &=\sum_{r=1}^\infty r  d_i^{(t)}(r) \\
    &= \sum_{r=1}^\infty r \bigl( 1-\gamma_{ii}^{(t+r-1)} \bigr) \prod_{j=1}^{r-1} \gamma_{ii}^{(t+j-1)} \\
%\end{align*}
% We can split up the infinite sum into partial sums for each period of length $L$, which gives
% \begin{align*}
    &=\sum_{k=0}^\infty \sum_{r = kL+1}^{kL+L} r \bigl( 1-\gamma_{ii}^{(t+r-1)} \bigr)  \prod_{j=1}^{r-1} \gamma_{ii}^{(t+j-1)}.
 \end{align*}
 Then, we see that in each summand with $k>0$ (from which follows $r>L$ in the inner sum) at least one \textit{full-length} product $\prod_{j=1}^L \gamma_{ii}^{(j)}$ is contained due to Equation (1). More specifically, a summand with $kL < r \leq (k+1)L$ contains $k$ full-length products, i.e.\ $(\prod_{j=1}^L \gamma_{ii}^{(j)})^k$ which is independent of $r$. Therefore, we obtain
\begin{align*}
    \mathbb{E} R_i^{(t)} = \dotsc = \sum_{k=0}^\infty \Bigl( \prod_{j=1}^L \gamma_{ii}^{(t+j-1)} \Bigr)^k \sum_{r = kL+1}^{kL+L} r \bigl( 1-\gamma_{ii}^{(t+r-1)} \bigr) \prod_{j=kL+1}^{r-1} \gamma_{ii}^{(t+j-1)}.
\end{align*}
Now, we see that again by virtue of Equation (1) the remaining product does not depend on the number of periods the chain stays in state $i$. Thus, by changing the indices, we arrive at
\begin{align*}
    \mathbb{E} R_i^{(t)} &= \dotsc \\
    &=\sum_{k=0}^\infty \Bigl( \prod_{j=1}^L \gamma_{ii}^{(t+j-1)} \Bigr)^k \sum_{r = 1}^{L} (kL +  r) \bigl( 1-\gamma_{ii}^{(t+r-1)} \bigr) \prod_{j=1}^{r-1} \gamma_{ii}^{(t+j-1)} \\
%\end{align*}
%Multiplying out, we obtain
%\begin{align*}
    &= \sum_{k=0}^\infty \Bigl( \prod_{j=1}^L \gamma_{ii}^{(t+j-1)} \Bigr)^k
    \Bigl( kL \sum_{r = 1}^{L} \bigl( 1-\gamma_{ii}^{(t+r-1)} \bigr) \prod_{j=1}^{r-1} \gamma_{ii}^{(t+j-1)} + \sum_{r = 1}^{L} r \bigl( 1-\gamma_{ii}^{(t+r-1)} \bigr) \prod_{j=1}^{r-1} \gamma_{ii}^{(t+j-1)} \Bigr)\\
    &= \Bigl( \sum_{r = 1}^{L} r \bigl( 1-\gamma_{ii}^{(t+r-1)} \bigr) \prod_{j=1}^{r-1} \gamma_{ii}^{(t+j-1)} \Bigr) \sum_{k=0}^\infty \Bigl( \prod_{j=1}^L \gamma_{ii}^{(t+j-1)} \Bigr)^k\\
    &+ L \Bigl( \sum_{r=1}^L \bigl( 1-\gamma_{ii}^{(t+r-1)} \bigr) \prod_{j=1}^{r-1} \gamma_{ii}^{(t+j-1)} \Bigr) \sum_{k=0}^\infty k \Bigl( \prod_{j=1}^L \gamma_{ii}^{(t+j-1)} \Bigr)^k\\
    &= \frac{\sum_{r=1}^L r \bigl( 1-\gamma_{ii}^{(t+r-1)} \bigr) \prod_{j=1}^{r-1} \gamma_{ii}^{(t+j-1)}}{1-\prod_{j=1}^L \gamma_{ii}^{(t+j-1)}}\\
    &+ L \Bigl( \sum_{r=1}^L \bigl( 1-\gamma_{ii}^{(t+r-1)} \bigr) \prod_{j=1}^{r-1} \gamma_{ii}^{(t+j-1)} \Bigr) \frac{\prod_{j=1}^L \gamma_{ii}^{(t+j-1)}}{ \bigl( 1-\prod_{j=1}^L \gamma_{ii}^{(t+j-1)} \bigr)^2}\\
    &= \frac{\sum_{r=1}^L r \bigl( 1-\gamma_{ii}^{(t+r-1)} \bigr) \prod_{j=1}^{r-1} \gamma_{ii}^{(t+j-1)}}{1-\prod_{j=1}^L \gamma_{ii}^{(t+j-1)}}
    + L \Bigl( 1 - \prod_{j=1}^L \gamma_{ii}^{(t+j-1)}\Bigr) \frac{\prod_{j=1}^L \gamma_{ii}^{(t+j-1)}}{\bigl( 1-\prod_{j=1}^L \gamma_{ii}^{(t+j-1)} \bigr)^2}\\
    &= \frac{\sum_{r=1}^L r d_i^{(t)}(r)+L \prod_{j=1}^L \gamma_{ii}^{(t+j-1)}}{1-\prod_{j=1}^L \gamma_{ii}^{(t+j-1)}}\\
    &= \frac{\sum_{r=1}^L r d_i^{(t)}(r) + L \bigl( 1- \sum_{r=1}^L d_i^{(t)}(r) \bigr)}{\sum_{r=1}^L d_i^{(t)}(r)}\\
    &= \frac{\sum_{r=1}^L r d_i^{(t)}(r) + L - L \sum_{r=1}^L d_i^{(t)}(r)}{\sum_{r=1}^L d_i^{(t)}(r)}\\
    &= \frac{L + \sum_{r=1}^L r d_i^{(t)}(r)}{\sum_{r=1}^L d_i^{(t)}(r)} -L,
\end{align*}
where in the fifth equality, we used the geometric sum and its derivative to calculate the infinite sums. Afterwards, we applied Lemma \ref{lem: sum_prod} in the second summand, simplified the fraction and applied Lemma \ref{lem: sum_prod} again.
\end{proof}

\newpage

\begin{proof}[Proof of Proposition 2.3]
    We need to obtain the overall dwell-time distribution as a mixture of the time-varying dwell-time distributions. For the weighting, we need to consider a random variable $\tau$ on the support $\{1, \dotsc, L\}$ that is uniformly distributed with $\Pr (\tau=t) = \frac{1}{L}$ for all $t = 1, \dotsc, L$. The realization of $\tau$ gives rise to every time point $t \in \{1, \dotsc, L\}$ considered in the mixture as a starting time for a dwell time. Then we can rewrite the time-varying dwell-time distributions as
    \begin{align*}
        d_i^{(t)}(r) = \Pr (S_{\tau+r} \neq i, S_{\tau+r-1} = i, \dotsc, S_{\tau+1} = i \mid S_{\tau} = i, S_{\tau-1} \neq i, \tau = t).
    \end{align*}
The probability we are interested in as the overall dwell-time distribution is 
    \begin{align*}
        \Pr (S_{\bigcdot+r} \neq i, S_{\bigcdot+r-1}=i, \dots, S_{\bigcdot+1}=i \mid S_{\bigcdot} = i, S_{\bigcdot-1} \neq i).
    \end{align*}
The dot notation is used to emphasize that this probability for the Markov chain to transition to state $i$ from any other state and to then stay in state $i$ $r$ times is unconditional of the time point in the cycle. We condition on the event that the transition $\neg i \rightarrow i$ has happened at some arbitrary time point in the cycle, prior to the stay.
Then, we can obtain the overall dwell-time distribution as
    \begin{align*}
    d_i(r) &= \Pr(S_{\bigcdot+r} \neq i, S_{\bigcdot+r-1} = i, \dotsc, S_{\bigcdot+1} = i \mid S_{\bigcdot} = i, S_{\bigcdot-1} \neq i) \\
    &= \frac{\Pr(S_{\bigcdot+r} \neq i, S_{\bigcdot+r-1} = i, \dotsc, S_{\bigcdot+1} = i, S_{\bigcdot} = i, S_{\bigcdot-1} \neq i)}{\Pr(S_{\bigcdot} = i, S_{\bigcdot-1} \neq i)} \\
    &= \frac{\sum_{t=1}^L \Pr(S_{\tau+r} \neq i, S_{\tau+r-1} = i, \dotsc, S_{\tau+1} = i, S_{\tau} = i, S_{\tau-1} \neq i, \tau = t)}{\Pr(S_{\bigcdot} = i, S_{\bigcdot-1} \neq i)} \\
    &= \sum_{t=1}^L \Pr(S_{\tau+r} \neq i, S_{\tau+r-1} = i, \dotsc, S_{\tau+1} = i \mid S_{\tau} = i, S_{\tau-1} \neq i, \tau = t) \\
    & \qquad \frac{\Pr(S_{\tau} = i, S_{\tau-1} \neq i, \tau = t)}{\Pr(S_{\bigcdot} = i, S_{\bigcdot-1} \neq i)}\\
    &= \sum_{t=1}^L d_i^{(t)}(r)  w_i^{(t)}.
    \end{align*}
 We now need to show that we can calculate the mixture weights explicitly to arrive at the mixture weights $ w_i^{(t)}$ as defined before, precisely that
    \begin{align*}
        w_i^{(t)}
        &= \frac{\Pr(S_{\tau} = i, S_{\tau-1} \neq i, \tau = t)}{\Pr(S_{\bigcdot} = i, S_{\bigcdot-1} \neq i)} \\
        &= \frac{\sum_{l \in \mathcal{S} \setminus i} \delta_l^{(t-1)} \gamma_{li}^{(t-1)}}{\sum_{t=1}^L \sum_{l \in \mathcal{S} \setminus i} \delta_l^{(t-1)} \gamma_{li}^{(t-1)}},
    \end{align*}
    where $\mathcal{S} = \{1, \dotsc, N\}$.
    We therefore consider the numerator and denominator separately. For the numerator, we need to consider all possible paths of the Markov chain from all states $l \neq i$ to state $i$:
    \begin{align*}
        \Pr(S_{\tau} = i, S_{\tau-1} \neq i, \tau = t) &= \sum_{l \in \mathcal{S} \setminus i} \Pr(S_{\tau} = i \mid S_{\tau-1} = l, \tau = t) \Pr(S_{\tau-1} = l \mid \tau = t) \Pr(\tau=t)\\
        &= \sum_{l \in \mathcal{S} \setminus i} \delta_l^{(t-1)} \gamma_{li}^{(t-1)} \frac{1}{L}.
    \end{align*}
    For the denominator, we obtain
    \begin{align*}
        \Pr(S_{\bigcdot} = i, S_{\bigcdot-1} \neq i) &= \sum_{t=1}^L \Pr(S_{\tau} = i, S_{\tau-1} \neq i, \tau = t) \\
        &= \frac{1}{L} \sum_{t=1}^L \sum_{l \in \mathcal{S} \setminus i} \delta_l^{(t-1)} \gamma_{li}^{(t-1)}.
    \end{align*}
Furthermore, we can calculate the mean of the distribution as
    \begin{align*}
    \mathbb{E} R_i &= \sum_{r=1}^\infty r d_i(r)\\
    &= \sum_{r=1}^\infty r \sum_{t=1}^L w_i^{(t)} d_i^{(t)}(r) \\
    % &= \sum_{r=1}^\infty \sum_{t=1}^L r w_i^{(t)} d_i^{(t)}(r) \\
    &= \sum_{t=1}^L w_i^{(t)} \sum_{r=1}^\infty r  d_i^{(t)}(r) \\
    &= \sum_{t=1}^L w_i^{(t)} \mathbb{E} R_i^{(t)}\\
    &= \sum_{t=1}^L w_i^{(t)} \left( \frac{L+\sum_{r=1}^L r d_i^{(t)}(r)}{\sum_{r=1}^L d_i^{(t)}(r)} - L\right) \\
    &= \sum_{t=1}^L w_i^{(t)} \frac{L+\sum_{r=1}^L r d_i^{(t)}(r)}{\sum_{r=1}^L d_i^{(t)}(r)} - L \sum_{t=1}^L w_i^{(t)}\\
    &= \sum_{t=1}^L \left( w_i^{(t)} \frac{L+\sum_{r=1}^L r d_i^{(t)}(r)}{\sum_{r=1}^L d_i^{(t)}(r)} \right) - L,
    \end{align*}
where the third equality is justified by the Fubini-Tonelli-theorem.
\end{proof}

\newpage

\subsection{Recursive calculation of the periodically stationary distribution}
\label{A1: Recursive}

Given a time-varying state distribution $\bm{\delta}^{(t)}$, for any $t \in \{1,\ldots,L\}$, as defined in Equation (3), the remaining $L-1$ stationary distributions can be calculated recursively. More formally, let $\bm{\delta}^{(t)}$ be the solution to $\bm{\delta}^{(t)} \Tilde{\bm{\Gamma}}_t = \bm{\delta}^{(t)}$. Then
    \begin{equation}
        \bm{\delta}^{(t+1)} = \bm{\delta}^{(t)} \bm{\Gamma}^{(t)}
    \end{equation}
    is the solution to
    $$\bm{\delta}^{(t+1)} \Tilde{\bm{\Gamma}}_{t+1} = \bm{\delta}^{(t+1)},$$ 
    since
    \begin{align*}
        \bm{\delta}^{(t+1)} \Tilde{\bm{\Gamma}}_{t+1} &= \bm{\delta}^{(t+1)} \bm{\Gamma}^{(t+1)} \bm{\Gamma}^{(t+2)} \dotsc \bm{\Gamma}^{(t+L)} \\
        &= \bm{\delta}^{(t)} \bm{\Gamma}^{(t)} \bm{\Gamma}^{(t+1)} \bm{\Gamma}^{(t+2)} \dotsc \bm{\Gamma}^{(t+L-1)} \bm{\Gamma}^{(t+L)} \\
        &= \bm{\delta}^{(t)} \Tilde{\bm{\Gamma}}_t \bm{\Gamma}^{(t+L)} = \bm{\delta}^{(t)} \bm{\Gamma}^{(t+L)} \\
        &= \bm{\delta}^{(t)} \bm{\Gamma}^{(t)} = \bm{\delta}^{(t+1)}.
    \end{align*}

For each of the $L$ time points, a naive approach involves forming the product of $L$ $N \times N$ matrices (requiring $L-1$ multiplications per point) and solving the linear system $\bm{\delta}^{(t)} \Tilde{\bm{\Gamma}}_t = \bm{\delta}^{(t)}$, requiring on the order of $N^3 L^2$ operations. In contrast, the recursive approach described above performs the full product and solves the system only for one time point, %, which requires on the order of $N^3 L$ operations. 
before computing the remaining $L-1$ results using $1 \times N$ vector by $N\times N$ matrix multiplications. %, requiting on the order of $(L-1)N^2$ operations. 
This yields a total of on the order of $N^3 L + (L-1)N^2$ operations. Both methods are generally minor compared to a single likelihood evaluation, which typically requires on the order of $N^2 T$ operations with $T \gg L$ and $N$ usually small.

\subsection{Memorylessness of the time-varying and overall dwell-time distributions}
\label{A2: Memorylessness}

\begin{proposition}
    Consider a periodically stationary Markov chain with period length $L$ as defined by Equations (1) and (3) and let $R_i^{(t)}$ denote the dwell time in state $i$ beginning at time $t$, following the time-varying dwell-time distribution as defined in Proposition 2.1. Then:
    \begin{enumerate}
        \item The time-varying dwell-time distribution satisfies a periodic memorylessness property in that
        $\Pr(R_i^{(t)}>L+s \mid R_i^{(t)}>L) = \Pr(R_i^{(t)} > s)$ for all $s \in \mathbb{N}$.
        \item The overall dwell-time distribution inherits this property.
    \end{enumerate}
\end{proposition}

\renewcommand{\thefigure}{9}
\begin{figure}
    \centering
    \includegraphics[width = 12cm]{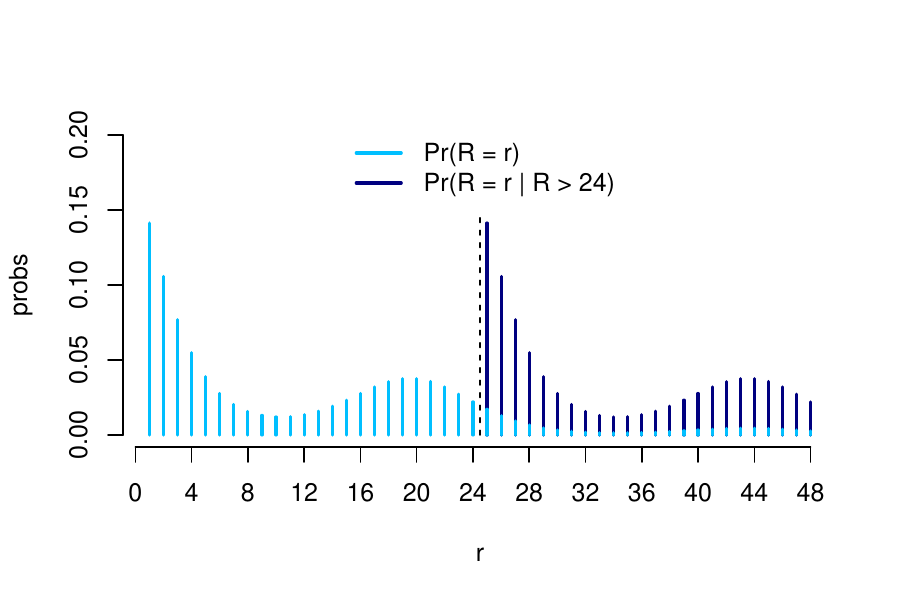}
    \caption{Visualization of the periodic memorylessness property of the overall dwell-time distribution. When standardizing $d_i(r)$ by the factor $\bigl(1-F_i(24)\bigr)^{-1}$ for $r = 25, ..., 48$, we again obtain $d_i(r), \; r = 1, \dotsc, 24$.}
    \label{fig:memorylessness}
\end{figure}

\begin{proof}[Proof of (1)]
Without loss of generality, we set $t=1$ and fix $s \in \mathbb{N}$. Then:
\begin{align*}
    \Pr(R_i^{(1)}>L+s \mid R_i^{(1)}>L) &= \frac{\Pr(R_i^{(1)}>L+s, R_i^{(1)}>L)}{\Pr(R_i^{(1)}>L)} \\
    &= \frac{\Pr(R_i^{(1)}>L+s)}{\Pr(R_i^{(1)} > L)}\\
    &= \frac{\Pr(R_i^{(1)}>L+s)}{1-\Pr(R_i^{(1)}\leq L)}\\
    &= \frac{\sum_{r=1}^\infty \bigl( 1-\gamma_{ii}^{(L+s+r)}\bigr) \prod_{j=1}^{L+s+r-1} \gamma_{ii}^{(j)}}{1-\sum_{r=1}^L \bigl( 1-\gamma_{ii}^{(r)}\bigr) \prod_{j=1}^{r-1} \gamma_{ii}^{(j)}}\\
    &= \frac{\bigl( \prod_{k=1}^L \gamma_{ii}^{(k)} \bigr) \sum_{r=1}^\infty \bigl( 1-\gamma_{ii}^{(L+s+r)}\bigr) \prod_{j=L+1}^{s+r-1} \gamma_{ii}^{(j)}}{\prod_{k=1}^L \gamma_{ii}^{(k)}}\\
    %&= \frac{\bigl( \prod_{k=1}^L \gamma_{ii}^{(k)} \bigr) \sum_{r=1}^\infty \bigl( 1-\gamma_{ii}^{(s+r)}\bigr) \prod_{j=1}^{s+r-1} \gamma_{ii}^{(j)}}{\prod_{k=1}^L \gamma_{ii}^{(k)}}\\
    &= \sum_{r=1}^\infty \bigl( 1-\gamma_{ii}^{(s+r)}\bigr) \prod_{j=1}^{s+r-1} \gamma_{ii}^{(j)}\\
    &= \Pr(R_i^{(1)} > s).
\end{align*}
In the fifth step, we see that in the numerator each summand contains one \textit{full-length product} $\prod_{k=1}^L \gamma_{ii}^{(k)}$ which we factor out, and apply Lemma \ref{lem: sum_prod} in the denominator. In the sixth step, we shift indices as $\gamma_{ii}^{(L+s)} = \gamma_{ii}^{(s)}$ (see Equation (1)) and simplify the fraction.
%hus we get
% \begin{align*}
%     \frac{\bigl( \prod_{k=1}^L \gamma_{ii}^{(k)} \bigr) \sum_{r=1}^\infty \bigl( 1-\gamma_{ii}^{(s+r)}\bigr) \prod_{j=1}^{s+r-1} \gamma_{ii}^{(j)}}{\prod_{k=1}^L \gamma_{ii}^{(k)}} = \sum_{r=1}^\infty \bigl( 1-\gamma_{ii}^{(s+r)}\bigr) \prod_{j=1}^{s+r-1} \gamma_{ii}^{(j)} = \Pr(R_i^{(1)} > s).
% \end{align*}
\end{proof}

\begin{proof}[Proof of (2)]
Again, we fix $s \in \mathbb{N}$, then:
\begin{align*}
    \Pr(R_i>L+s \mid R_i >L) &= \frac{\Pr(R_i>L+s)}{\Pr(R_i>L)} \\
    &= \frac{\Pr(R_i>L+s, R_i>L)}{\Pr(R_i>L)} \\
    &= \frac{\sum_{r=1}^\infty \sum_{t=1}^L w_i^{(t)} \bigl( 1-\gamma_{ii}^{(t+L+s+r-1)}\bigr) \prod_{j=1}^{L+s+r-1} \gamma_{ii}^{(t+j-1)}}{\sum_{r=1}^\infty \sum_{t=1}^L w_i^{(t)} \bigl( 1-\gamma_{ii}^{(t+L+r-1)}\bigr) \prod_{j=1}^{L+r-1} \gamma_{ii}^{(t+j-1)}} \\
    &= \frac{\sum_{t=1}^L w_i^{(t)} \sum_{r=1}^\infty \bigl( 1-\gamma_{ii}^{(t+L+s+r-1)}\bigr) \prod_{j=1}^{L+s+r-1} \gamma_{ii}^{(t+j-1)}}{\sum_{t=1}^L w_i^{(t)} \sum_{r=1}^\infty \bigl( 1-\gamma_{ii}^{(t+L+r-1)}\bigr) \prod_{j=1}^{L+r-1} \gamma_{ii}^{(t+j-1)}}.
\end{align*}
    We again realize that each summand in the numerator contains one \textit{full-length product} $\prod_{k=1}^L \gamma_{ii}^{(k)}$ which we factor out. In the denominator, we can rely on the Proof of (1) to rewrite each inner sum as
    \begin{align*}
        \sum_{r=1}^\infty \bigl( 1-\gamma_{ii}^{(t+L+r-1)}\bigr) \prod_{j=1}^{L+r-1} \gamma_{ii}^{(t+j-1)} = \Pr(R_i^{(t)}>L)
        = 1-\Pr(R_i^{(t)} \leq L)
        %&= 1-\sum_{r=1}^L \bigl( 1-\gamma_{ii}^{(t+r-1)}\bigr) \prod_{j=1} \gamma_{ii}^{(t+j-1)}
        = \prod_{k=1}^L \gamma_{ii}^{(k)}.
    \end{align*}
    As $\gamma_{ii}^{(t+L+s+r-1)} = \gamma_{ii}^{(t+r+s-1)}$ by Equation (1) we again shift indices. Thus:
    \begin{align*}
        \Pr(R_i>L+s \mid R_i >L) &= \frac{\bigl( \prod_{k=1}^L \gamma_{ii}^{(k)} \bigr) \sum_{t=1}^L w_i^{(t)} \sum_{r=1}^\infty \bigl( 1-\gamma_{ii}^{(t+s+r-1)}\bigr) \prod_{j=1}^{r+s-1} \gamma_{ii}^{(t+j-1)}}{\bigl( \prod_{k=1}^L \gamma_{ii}^{(k)} \bigr) \sum_{t=1}^L w_i^{(t)}} \\
        &= \sum_{r=1}^\infty \sum_{t=1}^L w_i^{(t)} \bigl( 1-\gamma_{ii}^{(t+s+r-1)}\bigr) \prod_{j=1}^{r+s-1} \gamma_{ii}^{(t+j-1)} \\
        &= \Pr(R_i > s).
    \end{align*}
    In the second step, we simplified the fraction and realized that the sum of the weights equals one.
\end{proof}

\newpage

\subsection{Additional figures}
\label{A: figures}

\renewcommand{\thefigure}{10}
\begin{figure}[!htb]
    \centering
    \includegraphics[width=16cm]{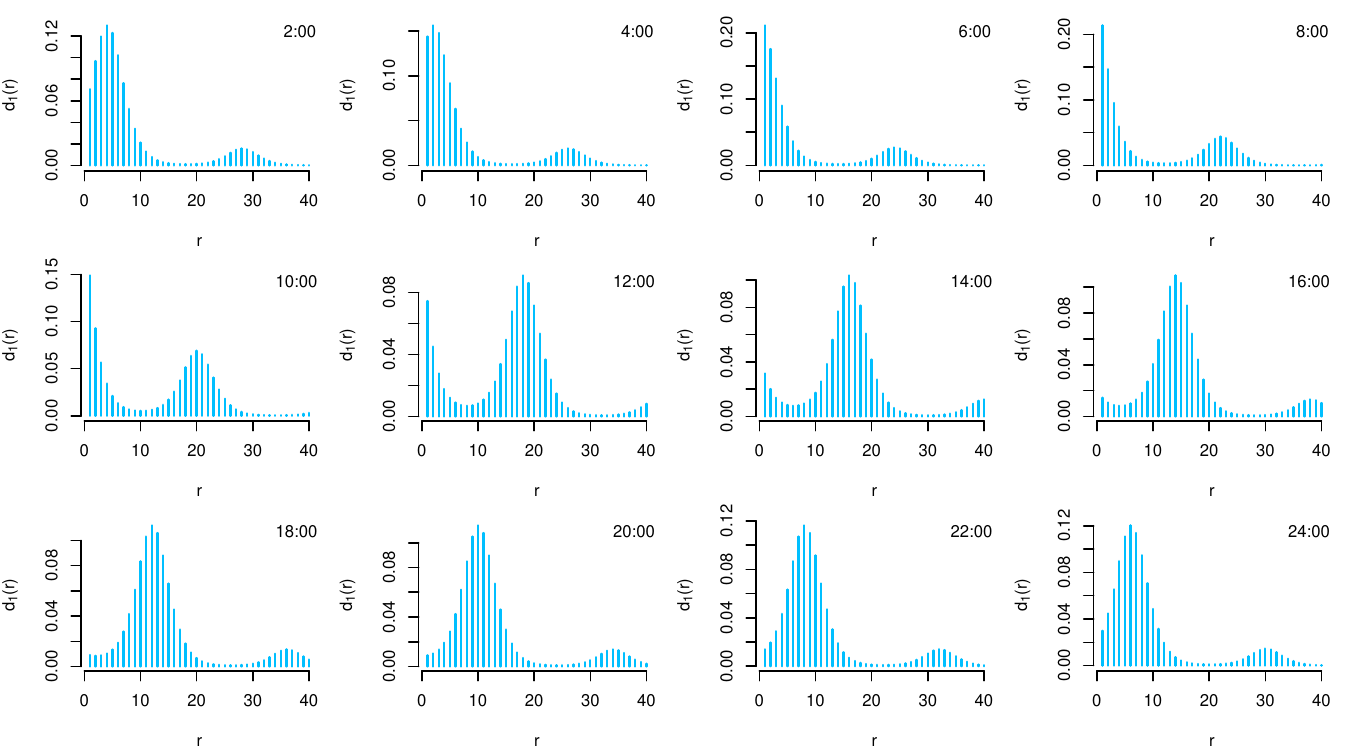}
    \caption{Time-varying dwell-time distribution (state 2) of an example 2-state HMM with trigonometric modeling of periodic variation (see Appendix \ref{A: parameters} for the precise model configuration).}
    \label{fig: timevarying}
\end{figure}

\renewcommand{\thefigure}{11}
\begin{figure}
    \centering
    \includegraphics[width=14cm]{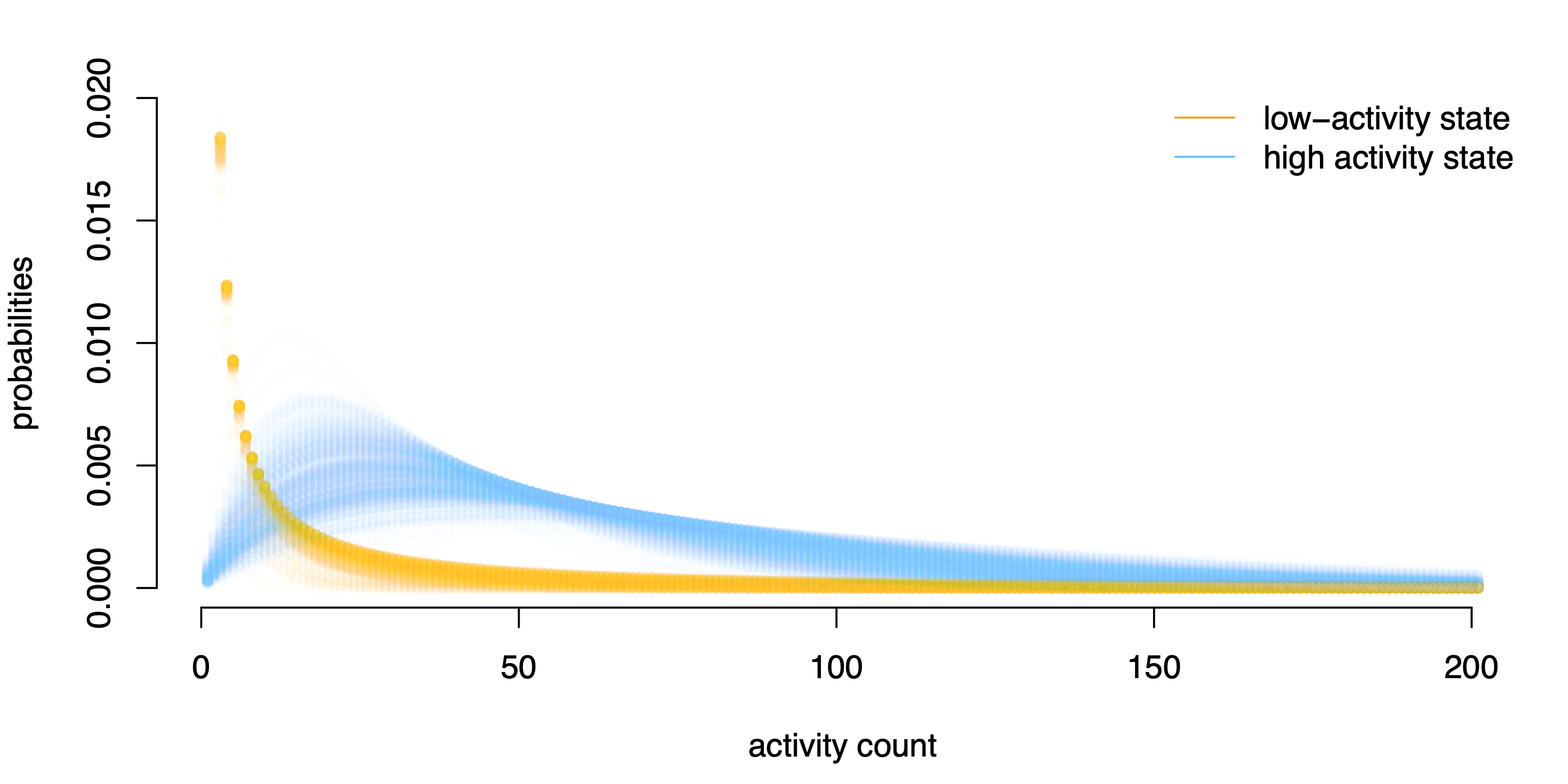}
    \caption{Weighted state-dependent distributions. 200 draws from the distribution of the random mean.}
    \label{fig: state_dep}
\end{figure}

\renewcommand{\thefigure}{12}
\begin{figure}
    \centering
    \includegraphics[width=14cm]{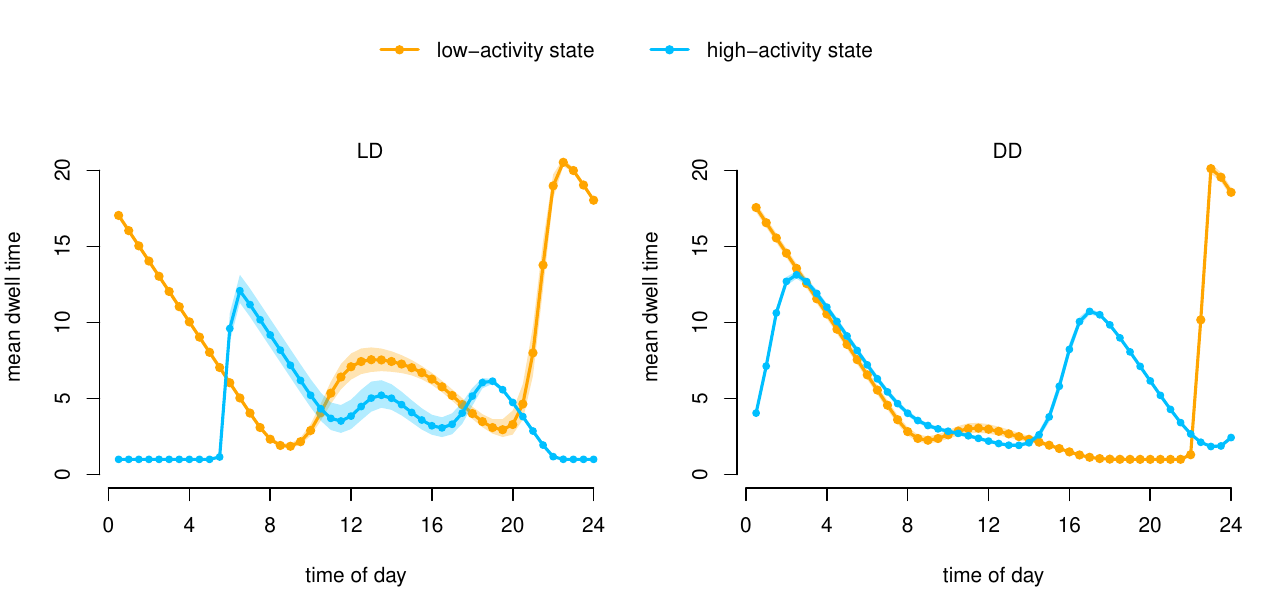}
    \caption{Mean dwell times as a function of the time of day, for LD and DD condition. Pointwise confidence intervals were obtained via Monte Carlo simulation from the approximate normal distribution of the maximum likelihood estimator.}
    \label{fig: mean dwell times}
\end{figure}

\renewcommand{\thefigure}{13}
\begin{figure}
    \centering
    \includegraphics[width=14cm]{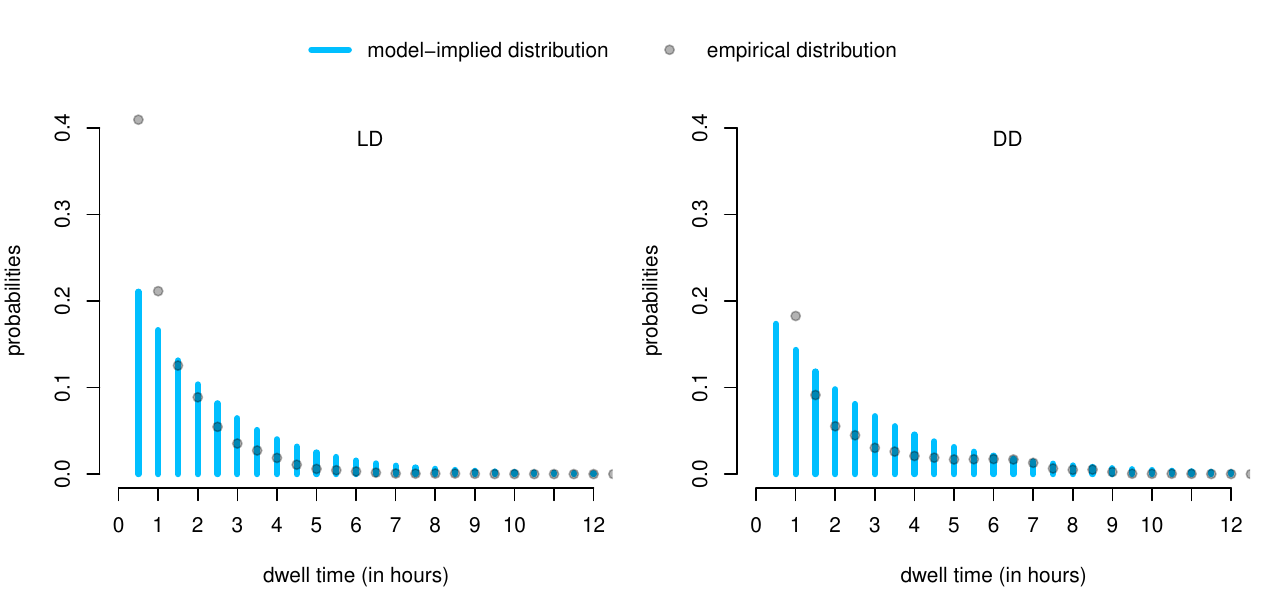}
    \caption{Overall dwell-time distribution of the high-activity state for LD and DD condition, analytically (blue bars) and empirically (gray dots) derived from the fitted homogeneous HMM.}
    \label{fig: overall dwell times hom}
    \vspace{5in}
\end{figure}

\newpage

\subsection{Parameters used for Figures 2, 3, 4, 9, and 10}
\label{A: parameters}

Figures 2, 3, 4, 9 and 10 were 
generated based on 2-state HMMs with the entries of the t.p.m.\ modeled as specified in Equation (2). The parameter values used for Figure 2 are
$$
\bm{\beta}^{(12)} = (-2, -1, -1)\quad \text{and} \quad
\bm{\beta}^{(21)} = (-2, 2, 2),
$$
for scenario 1,
$$
\bm{\beta}^{(12)} = (-5, -1, -1)\quad \text{and} \quad
\bm{\beta}^{(21)} = (-5, 1, 1),
$$
for scenario 2, and
$$
\bm{\beta}^{(12)} = (-3, -0.5, -1, 1, -2)\quad \text{and} \quad
\bm{\beta}^{(21)} = (-3, -0.5, 2, 0.5, -0.5),
$$
for scenario 3.

The parameter values for Figure 3 are
$$
\bm{\beta}^{(12)} = (-1.2, 0.85, 0.15)\quad \text{and} \quad
\bm{\beta}^{(21)} = (-1.5, -0.7, -1.3),
$$
those for Figures 4, 9 and 10 are
$$
\bm{\beta}^{(12)} = (-3,1.5,-0.9) \quad \text{and} \quad
\bm{\beta}^{(21)} = (-3,1.2,-1.1).
$$

\end{document}